\documentclass[a4paper,USenglish,numberwithinsect,autoref, thm-restate]{lipics-v2021}
\usepackage{amsmath,amssymb}
\usepackage{graphicx}
\usepackage{latexsym}
\usepackage{xspace}  % allows for indexgeneration
\usepackage{makeidx}  % allows for indexgeneration
\usepackage{algorithm}          % For presenting algorithms
\usepackage[noend]{algorithmic}
\usepackage{bm}
\usepackage{blkarray}
\usepackage{multirow}


\def\mO{{\mathrm O}}
\def\noi{\noindent}
\def\sm{\setminus}
\def\med{\medskip}

\def\depart{\varphi}
\def\size{{\it size}}
\def\transit{\tau}
\def\transituni{\mu}
\def\capa{c}
\def\source{s}
\def\sink{t}
\def\W{W}
\def\P{P}
\def\Q{\rho}

\def\RA{RP}
\def\opt{\mathsf{opt}}
\def\tt{\theta}
\def\T{T}

\def\I{{\cal I}}
\def\LP{LP}

\newcommand\pro{product\xspace}
\newcommand\prods{products\xspace}
\newcommand\house{warehouse\xspace}
\newcommand\houses{warehouses\xspace}
\newcommand\outdegree{carry-out\xspace}
\newcommand\indegree{carry-in\xspace}
\newcommand\mainpro{reallocation\ problem\xspace}
\newcommand\dmgraph{demand graph\xspace}
\newcommand{\alg}[1]{\textsc{#1}}   % Format for names of algorithms

\newcommand{\nop}[1]{}

\newcommand{\ishii}[1]{{\color{black} #1}}

\title{Reallocation Problems with Minimum Completion Time} %TODO Please add

\author{Toshimasa Ishii}{Graduate School of Economics and Business
Administration, Hokkaido University, Sapporo, 060-0809,
Japan}{ishii@econ.hokudai.ac.jp}{}{}
\author{Jun Kawahara
}{Graduate School of Informatics, Kyoto University, Kyoto, 606-8501, Japan}
{jkawahara@i.kyoto-u.ac.jp}
{}{}
\author{Kazuhisa Makino
}{RIMS, Kyoto University, Kyoto, 606-8502, Japan}
{makino@kurims.kyoto-u.ac.jp}{}{}

\author{Hirotaka Ono
}{Graduate School of Informatics, 
  Nagoya University, Nagoya 464-8601, Japan}
{ono@nagoya-u.jp}{}{}

\authorrunning{T. Ishii et al.} %TODO mandatory. First: Use abbreviated first/middle names. Second (only in severe cases): Use first author plus 'et al.'

\Copyright{Toshimasa Ishii, Jun Kawahara, Kazuhisa Makino and
Hirotaka Ono} %TODO mandatory, please use full first names. LIPIcs license is "CC-BY";  http://creativecommons.org/licenses/by/3.0/

 %TODO mandatory: Please choose ACM 2012 classifications from https://dl.acm.org/ccs/ccs_flat.cfm 

 %TODO mandatory; please add comma-separated list of keywords

\category{} %optional, e.g. invited paper

\relatedversion{} %optional, e.g. full version hosted on arXiv, HAL, or other respository/website
\acknowledgements{This work was partially supported by the joint project
of Kyoto University and Toyota Motor Corporation, titled ^^ ^^ Advanced
Mathematical Science for Mobility Society''.}%optional

\nolinenumbers %uncomment to disable line numbering

\hideLIPIcs  %uncomment to remove references to LIPIcs series (logo, DOI, ...), e.g. when preparing a pre-final version to be uploaded to arXiv or another public repository

\EventEditors{John Q. Open and Joan R. Access}
\EventNoEds{2}
\EventLongTitle{42nd Conference on Very Important Topics (CVIT 2016)}
\EventShortTitle{CVIT 2016}
\EventAcronym{CVIT}
\EventYear{2016}
\EventDate{December 24--27, 2016}
\EventLocation{Little Whinging, United Kingdom}
\EventLogo{}
\SeriesVolume{42}
\ArticleNo{23}
\begin{document}

\maketitle

\begin{abstract}
Reallocation scheduling is one of the most fundamental problems in various areas such as supply chain management, logistics, and transportation science. In this paper, we introduce the reallocation problem that models the scheduling in which products are with fixed cost, non-fungible, and reallocated  in parallel, 
and comprehensively study the complexity of 
the problem under various settings of the transition time, product size, and capacities. 
We show that  the problem can be solved in polynomial time
 for a fundamental setting where the product size and transition time  are both uniform.
 We also show that the feasibility of the problem is NP-complete even for little more general settings, which implies that no polynomial-time algorithm constructs a feasible schedule of the problem unless P$=$NP. 
% inapproximability of the problem, i.e., no algorithm exists with approximation guarantee. 
 We then consider the relaxation of the problem, which we call the capacity augmentation, 
 and derive a reallocation schedule feasible with the augmentation such that 
 the completion time is at most the optimal of the original problem. 
 When the warehouse capacity is sufficiently large, 
  we design constant-factor approximation algorithms under all the settings. 
  We also show the relationship between the reallocation problem and 
  the bin packing problem when the warehouse and carry-in capacities are sufficiently large. 
\end{abstract}

%\newpage

\section{Introduction}\label{intro:sec}

{\large\bf Problem setting}: Suppose that there are several warehouses that store many products (or items). 
Some of the  products are already stored at the designated warehouses, and the others are stored at tentative warehouses. Such temporally stored  products should be reallocated to designated warehouses. Namely, each  product $p$ at the tentative warehouse $s(p)$ is required to be reallocated to the designated warehouse $t(p)$. To reallocate  product $p$, it takes a certain length of time $\tau(p)$, called transit time of $p$.  
Each  product $p$ also has size $size(p)$,  and each warehouse has three kinds of capacities, that is, (1) the capacity of warehouse itself,  (2) carry-in size capacity, and (3) carry-out size capacity. Capacity (1) restricts the total size of  products stored in a warehouse at each moment. Capacities (2) and (3) restrict the total size of  products that are simultaneously carried in and out, respectively. In this setting, we consider the problem of finding a reallocation schedule with minimum completion time. 
 
As an illustrative  example of our problem, let us consider the following scenario. 
There are 6  products $p_1, p_2, p_3, p_4, p_5$ and $p_6$ of sizes $1, 3, 5, 6, 3$ and $4$, respectively. 
Two warehouses $W_1$ and $W_2$ have the warehouse capacities $20$ and $10$, carry-in capacities $6$ and $5$, and carry-out capacities $5$ and $5$, respectively. 
The transit time satisfies  $\tau(p_i)=1$ for every $i$ except $i=6$, and $\tau(p_6)=2$.  
Initially,  products $p_1, p_2, p_3$ and $p_4$ are stored in $W_1$, 
which satisfies the warehouse capacity constraint, since  their total size $15$ is smaller than the warehouse capacity $20$ of $W_1$.
Products $p_5$ and $p_6$ are initially stored in $W_2$, which also satisfies the warehouse capacity constraint. Suppose that $p_1, p_2$ and $p_3$ are designated to be stored in $W_2$, whereas $p_4$, $p_5$ and $p_6$ are designated to be stored in $W_1$. 
Note that all  products can be stored in the designated warehouses, since the designated allocation satisfies the warehouse capacity constraint. 
%The total size $13$ of $p_4, p_5, p_6$ is less than 20 the capacity of $W_1$, 
%and  the total size $9$ of $p_1,p_2,p_3$ is also less than the capacity $10$ of $W_2$; the total capacities of $W_1$ and $W_2$ are satisfied in the designated warehouses. 
We also note that $p_1,p_2$ and $p_3$ cannot be moved simultaneously, since the carry-out capacity constraint is violated. Figure \ref{fig:example} depicts the initial and target configurations of this example.  
\begin{figure}[hbtp]
%  \begin{minipage}[b]{0.40\linewidth}
%    \centering
%    \includegraphics[keepaspectratio, width=60mm]{fig/warehouse1.png}
%    \caption{Initial configuration} \label{fig:example1}
%  \end{minipage}
%  \hspace*{1cm}
%  \begin{minipage}[b]{0.40\linewidth}
%    \centering
%     \includegraphics[keepaspectratio, width=60mm]{fig/warehouse2.png}
%    \caption{Target configuration} \label{fig:example2}
%  \end{minipage}
    \centering
     \includegraphics[keepaspectratio, width=130mm]{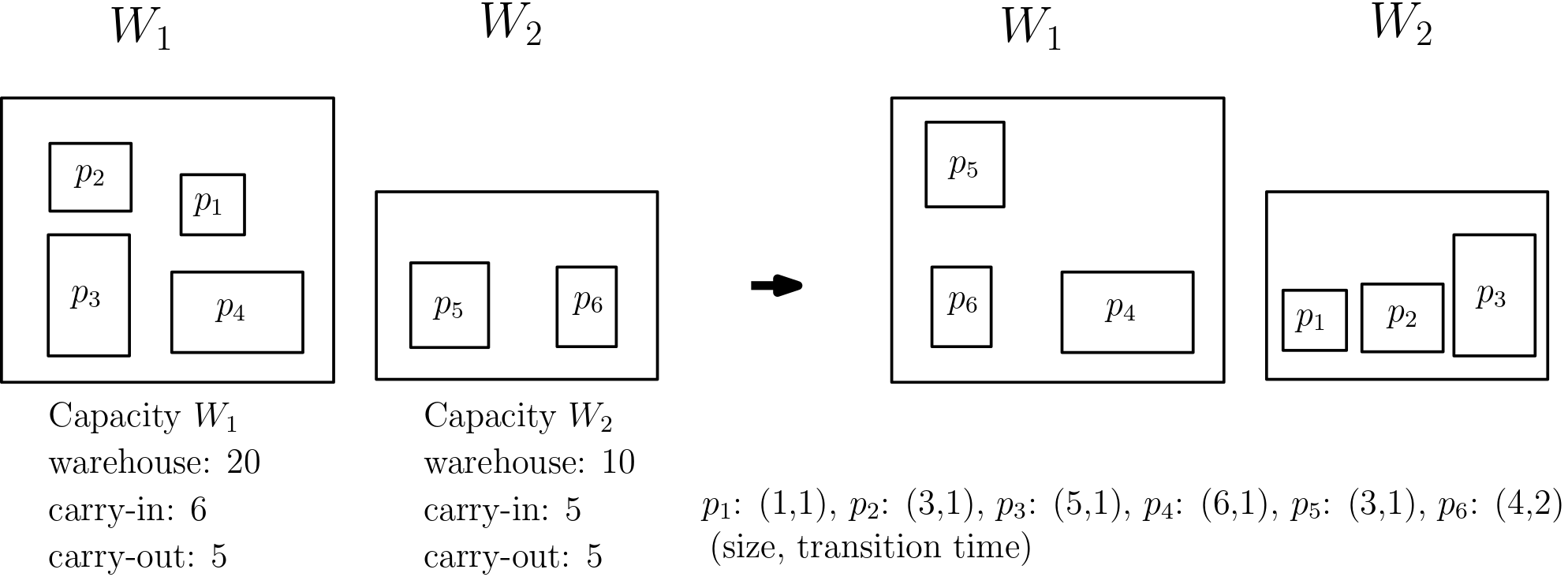}
     \caption{Example: initial (left) and target (right) configurations.} \label{fig:example}
\end{figure}

In this example, 
%a possible schedule is to carry out $p_1$ and $p_2$ from $W_1$, and $p_6$ from $W_2$ at time $0$, 
we need to move $p_1$, $p_2$ and $p_3$ to $W_2$, which should be done separately due to carry-out constraint $5$. Thus we consider to move $p_1$ and $p_2$ at time $0$ and then move $p_3$ at time $1$, for example. In such a case, we also need to move either $p_5$ or $p_6$ to $W_1$ at time $0$, otherwise it violates capacity $10$ of warehouse $W_2$ at time $1$. Thus a possible schedule is to carry out $p_1$ and $p_2$ from $W_1$, and $p_5$ from $W_2$ at time $0$, and then carry out $p_3$ from $W_1$ and $p_6$ from $W_2$ at time $1$. The completion time is $3$ in this case, because $\tau(p_6)=2$ and thus $p_6$ reaches $W_1$ at time $3$. One may consider that it is better to carry out $p_6$ instead of $p_5$ at time $0$ and then carry out $p_5$ at time $1$, but it violates carry-in constraint $6$ of $W_1$ at time $2$. By this observation, we can see that the minimum completion time is $3$ in this example. 
%----------------
%I hope I did not eliminate the sentences
%
%Please add it <= Ono sensei
%
%--------------------------

\medskip
\noindent
{\large\bf Applications and related work}: 
%\subsection*{Applications and Related work}
Reallocation scheduling is one of the most fundamental problems in various areas such as supply chain management, logistics, and transportation science.  Many models and variants of reallocation have been studied from both  theoretical and practical viewpoints. 
In fact, our problem setting is initiated by an industry-academia joint project of Advanced Mathematical Science for Mobility Society by Toyota Motor Corporation and Kyoto University~\cite{Toyota}.     
The reallocation models and their variants are categorized by the following aspects: (1) reallocation cost, (2) fungibility of  products, and (3) parallel/sequential execution.  
%(1) a cost of reallocating  products is a matter, or not, (2)  products are fungible or not, (3) reallocation can be done in a parallel way or only in a sequential way. (1) 
Our problem assumes that (1) a cost (i.e., transit time) of reallocating a product is given in advance, (2)  products are not fungible, 
and (3) reallocations are done in a parallel way, but many other settings are possible. 

For example, the dial-a-ride problem is regarded as the vehicle routing problem for reallocation~\cite{cordeau2007dial,ho2018survey}, which designs vehicle routes and schedules for customers  who request the pick-up and drop-off points. 
In the dial-a-ride problem, (1) the cost of reallocation depends on the routes, (2) customers are not fungible of course, (3) reallocation is done in a sequential way for one vehicle. 
The bike-sharing rebalancing problem is  to schedule trucks for re-distributing shared bikes with the objective of minimizing total cost~\cite{dell2014bike}. In the bike-sharing rebalancing problem, (1) the cost of reallocation also depends on the routes of trucks and (3) reallocation is done in a sequential way  for one truck again, but (2) shared bikes are fungible, i.e.,  the desired bike distribution is created without distinguishing between bikes.  
%In fact, the reallocation schedule models real-world problems discussed in a practically oriented project with Toyota Motor Corporation, which include reallocating shared cars in parking lots.    
%As an industry-academia joint project, Toyota Motor Corporation and Kyoto University recently launched a research unit of Advanced Mathematical Science for Mobility Society~\cite{Toyota}. In the project, several real-world problems other than above are realized to be modeled as reallocation scheduling, e.g., reallocating shared cars in parking slots. 
%Reallocating shared cars in parking slots is a similar problem, which is considered in an industry-academia joint project of Advanced Mathematical Science for Mobility Society by Toyota Motor Corporation and Kyoto University~\cite{Toyota}. In the problem, each car directly moves to another parking slot, whereas shared bikes are delivered by trucks along their routes. In the model, (2) cars are fungible in many cases, and thus (1) the reallocation cost depends only on the assignment of cars to the parking slots, and (3) reallocation is done in a parallel way.  

In these problems, the costs arise from rather transportation of vehicles for delivery than reallocation of  products.  
To investigate the nature of reallocation itself, it might be reasonable to assume that a cost (i.e., transit time) of reallocating a product is given in advance.   
%Not only that, the reallocation scheduling appears in various fields, such as memory allocation with limited buses and settlement fund circulation.  
Miwa and Ito \cite{miwa1996complexity,miwa1997linear,miwa2000np} focus on reallocation scheduling  under the setting: (1) a cost of reallocating  products is uniform, (2)  products are not fungible, 
and (3) reallocations are done in a sequential way, which is different from ours. 
%Miwa and Ito \cite{miwa1996complexity,miwa1997linear,miwa2000np} investigate the computational complexity of another type of reallocation scheduling, which computes how to move vacant spaces. 
Although they do not deal with non-uniform transition time, several intractability results as well as polynomial solvability ones  were obtained \cite{miwa1996complexity,miwa1997linear,miwa2000np}.  
A  problem similar to the model of Miwa and Ito was also considered in \cite{HAYAKAWA201986}, in the context of fund circulation. In \cite{hayakawa2018does,hayakawa2020liquidity}, Hayakawa pointed out the importance of controlling a payment ordering among banks in terms of the stability of economics, and introduced the problem to find a payment ordering that minimizes extra money to put in order to make payments without shortage, which can be viewed as a  reallocation scheduling problem. Here 
 the extra money corresponds to vacant spaces for the reallocation. 
From such a viewpoint, the computational complexity of the fund circulation was investigated in \cite{HAYAKAWA201986}.

\medskip
\noindent
{\large\bf Our contribution}: 
%\subsection*{Our contribution} 
In this paper, we consider the reallocation scheduling in which  products are with fixed cost,  non-fungible, and reallocated  in parallel, which we simply call the {\em reallocation problem}. %where it is simply called the {\em reallocation problem}. 
We investigate the computational complexity of the reallocation problem under various scenarios. 
We first see the most basic scenario  where both the  product size and transition time are uniform. In this scenario, we present an $\mO(\ishii{mn} \log m)$-time algorithm to find a reallocation schedule with minimum completion time, where
$m$ and $n$ denote the number of  products and warehouses, respectively.
%$m$ denotes the number of  products. %More precisely, if $\tau(p)=1$ for every  product $p$, 
The algorithm utilizes a cycle decomposition for the so-called demand graph. 
When the  product size is only uniform, the problem turns out to be NP-hard, and we show that the algorithm for the basic scenario above provides a reallocation schedule 
whose completion time is at most twice the optimal, if one of the carry-out and carry-in capacities is sufficiently large in addition.
Here we note that the carry-in and carry-out capacity constraints are always %automatically 
satisfied if the carry-in and carry-out capacities are sufficiently large, respectively. 
For more general scenarios, even the feasibility is NP-complete for very restricted settings: (1) we have only two warehouses and (2) the transit time  and  the  product size satisfy
$\transit(p)=1$
and $\size(p) \in \{1,2\}$ for every  product $p$.
Due to the hardness of the feasibility, 
%we have no choice but to give up constructing %approximation algorithms for the reallocation problem. 
%polynomial-time algorithms with approximation guarantee for the problem as it is.    
no polynomial-time algorithm can construct a feasible schedule of the problem unless P$=$NP.  
Instead, 
we admit the relaxation of capacity constraints, which we call \emph{capacity augmentation}. 
Namely, we augment the original capacities with additional ones, and try to find a reallocation schedule such that it is feasible with the augmented instance and has small completion time. 
By utilizing the  bi-criteria algorithm schemes
%iterative approximation algorithms
 for the generalized assignment problem \cite{ST93} and the $2$-sided placement problem \cite{KMRT15},  
we can find in polynomial time  a reallocation schedule such that 
(i) the completion time is at most the minimum completion time for the original problem, (ii)  warehouses can be stored at most twice of the original capacities (i.e., warehouse capacity $c(w)$ is augmented with $c(w)$ for each warehouse $w$), 
and (iii) carry-out/in capacities are enlarged to the original ones plus the largest and second largest carried-out/in product sizes (i.e., the carry-out and carry-in capacities  for each warehouse $w$ are respectively augmented with 
$\sigma^+_1(w)+\sigma^+_2(w)$ and $\sigma^-_1(w)+\sigma^-_2(w)$, where $\sigma^+_i(w)$ and $\sigma^-_i(w)$ denote the $i$-th largest size of  products that are initially and finally stored in warehouse $w$, respectively). 
In the scenario when the carry-in (resp., carry-out) capacity is sufficiently large, condition (iii) is strengthened into the original carry-out (resp., carry-in) capacity plus the largest size of carry-out (resp., carry-in)  product size. 
We remark that if both carry-in and carry-out capacities are sufficiently large, the problem turns out to be trivial, where an optimal reallocation schedule can be obtained by sending all  products in time $0$. 
 Table \ref{tab:result0} summarizes the results above, i.e., for general warehouse capacity. 
\begin{table}[thbp]
    \caption{The summary of results for general warehouse capacity.  } 
    \scriptsize
    \label{tab:result0}
    \begin{center}
    \begin{tabular}{|c|c|c|c||c|c|}\hline 
     product & transition & \multicolumn{2}{c||}{capacity} & \multirow{2}{*}{complexity $\&$} & \multirow{3}{*}{capacity augmentation$^\dagger$}\\ \cline{3 - 4}   
     size & time  & carry-out & carry-in& \multirow{2}{*}{approximability$^*$}&\\ 
     ($size(p)$) & ($\tau(p)$) & ($d^+$) & ($d^-$)  & &\\  \hline \hline 
    \multirow{2}{*}{uniform} & \multirow{2}{*}{uniform}  & \multirow{2}{*}{general} & \multirow{2}{*}{general} & $\mO(\ishii{mn \log m})$ & \multirow{2}{*}{---------}\\ 
    &&&&[Th. \ref{poly:th}]&\\\hline
    \multirow{2}{*}{uniform} & \multirow{2}{*}{general}  & general & $+\infty$ &  2-apx &\multirow{2}{*}{---------}\\  
      &  & $(+\infty$ & general) & [\ishii{Th. \ref{2apx:th}}]&\\ \hline
   \multirow{2}{*}{uniform} & \multirow{2}{*}{general} & \multirow{2}{*}{general} & \multirow{2}{*}{general} &  NPH &\multirow{2}{*}{---------}\\ 
   &&&&[\ishii{Th \ref{flowshop:th}}]&\\ \hline
      \multirow{2}{*}{general} & \multirow{2}{*}{general}    & \multirow{2}{*}{$+\infty$} & \multirow{2}{*}{$+\infty$} & $\mO(m+n)$& \multirow{2}{*}{---------} \\ 
     &&&&[\ishii{Th. \ref{trivial:th}}]&\\\hline
   &&&&&$\hat{c}=2c$\\
   \multirow{2}{*}{general} & \multirow{2}{*}{uniform} & general &$+\infty$ & NPC &$\hat{d}^+=d^++\sigma^+_1$\\  
      &  & $(+\infty$ & general) & [Th. \ref{hard1:th}]&($\hat{d}^-=d^-+\sigma^-_1$)\\ 
      &&&&&[Th.~\ref{bicriteria:th} (ii)]\\ \hline
    &&&&&$\hat{c}=2c$ \\
    \multirow{2}{*}{general} & \multirow{2}{*}{uniform} & \multirow{2}{*}{general} & \multirow{2}{*}{general} & NPC &$\hat{d}^+=d^++\sigma^+_1+\sigma^+_2$\\ 
      &  & & &[Th. \ref{hard1:th}]&$\hat{d}^-=d^-+\sigma^-_1+\sigma^-_2$\\ 
      &&&&&[Th.~\ref{bicriteria:th} (i)]\\
      \hline
    &&&&&$\hat{c}=2c$\\
   \multirow{2}{*}{general} & \multirow{2}{*}{general}  & general & $+\infty$ & 
   NPC &$\hat{d}^+=d^++\sigma^+_1$\\  
      &  & ($+\infty$ & general) & [Th. \ref{hard1:th}]&($\hat{d}^-=d^-+\sigma^-_1$)\\ 
      &&&&&[Th.~\ref{bicriteria:th} (ii)]\\\hline
     & & &   & &$\hat{c}=2c$ \\ 
 \multirow{2}{*}{general}& \multirow{2}{*}{general}&\multirow{2}{*}{general}&\multirow{2}{*}{general}&NPC&$\hat{d}^+=d^++\sigma^+_1+\sigma^+_2$ \\
    &&&&[Th. \ref{hard1:th}]&$\hat{d}^-=d^-+\sigma^-_1+\sigma^-_2$\\
    &&&&&[Th.~\ref{bicriteria:th} (i)]\\ \hline  
    \end{tabular}
    \end{center}
    $^*$: ``NPC'' stands for  NP-completeness of the feasibility of the reallocation problem and ``NPH'' stands for NP-hardness for the reallocation problem.\\
    $^\dagger$: $\hat{c}$, $\hat{d}^+$, and $\hat{d}^-$ respectively denote augmented warehouse, carry-out, and carry-in capacities.
\end{table}
\begin{table}[hbtp]
    \caption{The summary of results for 
    %uniform
    \ishii{sufficiently large}
    warehouse capacity.  } 
    \small
    \label{tab:result1}
    \begin{center}
    \begin{tabular}{|c|c|c|c||c|}\hline 
     product &  transition & \multicolumn{2}{c||}{capacity} & \multirow{3}{*}{complexity $\&$ approximability$^*$} \\ \cline{3 - 4}   
     size & time  & carry-out & carry-in& \\ 
     ($size(p)$) & ($\tau(p)$)  & ($d^+$) & ($d^-$)  & \\  \hline \hline 
   \multirow{2}{*}{uniform} & \multirow{2}{*}{general} & general & $+\infty$ & \multirow{2}{*}{\ishii{$\mO(n+m\log{m})$}  [\ishii{Th. \ref{apx:th}(i)}]}\\  
      &  & $(+\infty$ & general) & 
   \\ \hline
   \multirow{2}{*}{uniform} & \multirow{2}{*}{general}  & \multirow{2}{*}{general} & \multirow{2}{*}{general}  & NPH  [\ishii{Th \ref{flowshop:th}}] \\ 
   &&&&\ishii{4-apx} [Th. \ref{apx:th}\ishii{(iv)}] \\ \hline
   \multirow{2}{*}{general} & \multirow{2}{*}{uniform}  & general & $+\infty$ & NPH [Th. \ref{hard2:th}]\\  
      &  & $(+\infty$ & general) &  3/2-apx  (tight) [Ths. \ref{apx:th}\ishii{(ii)},\ref{binpack1:th}]\\ \hline
    \multirow{2}{*}{general} & \multirow{2}{*}{uniform} & \multirow{2}{*}{general} & \multirow{2}{*}{general}  & NPH [Th. \ref{hard2:th}]\\
    &&&& 6-apx [Th. \ref{apx:th}\ishii{(v)}]\\ \hline
 \multirow{2}{*}{general} & \multirow{2}{*}{general} & general & $+\infty$ & NPH [Th. \ref{hard2:th}] \\  
      &  & $(+\infty$ & general) &  7/4-apx [Th. \ref{apx:th}\ishii{(iii)}] \\ \hline
    \multirow{2}{*}{general} & \multirow{2}{*}{general} & \multirow{2}{*}{general} & \multirow{2}{*}{general}  & NPH [Th. \ref{hard2:th}]\\
    &&&& 6-apx [Th. \ref{apx:th}\ishii{(v)}]\\ \hline 
    \end{tabular}
    \end{center}
    $^*$: ``NPH'' stands for NP-hardness for the reallocation problem. 
\end{table}
%We also show that our problem is NP-hard, even if the  product size is uniform, 
%which together with the positive result and (2) mentioned above  provides a sharp boundary for the complexity  of our problem. 
%Namely, it is polynomially solvable if the transit time and  product size are both uniform, 
%and it NP-hard if at least one of them is not uniform.  

We finally consider the scenario when all warehouses have sufficiently large capacities, where the summary of our results can be found in Table \ref{tab:result1}.  
In this setting, we propose a $6$-approximation algorithm for our problem that transforms a {\em relaxed} schedule above %the above  bi-criteria-type algorithm
into a feasible schedule such that the completion time is at most 6 times of the minimum completion time. 
In the setting when  at least one of  carry-in and  carry-out capacities is sufficiently large in addition, we present a  $7/4$-approximation algorithm that employs as a subroutine the first-fit-decreasing algorithm for the bin packing problem. 
If we further assume that the transition time is uniform, the approximation ratio is improved to $3/2$, which is best possible in the setting. 
This follows from the fact that the problem in this setting  is essentially equivalent to the bin packing problem. 
We also show that the reallocation problem can be solved in polynomial time, if the  product size is uniform and at least one of the carry-in and carry-out capacities are sufficiently large in addition. 

The rest of the paper is organized as follows. In Section \ref{preliminaries:sec}, 
%we formally define the problem and give a basic observation about the capacity constraints. 
we give formal definitions and basic observations. 
Section \ref{size2:sec}  presents a polynomial-time algorithm for the setting where both  size and transition time of products are uniform. 
Section \ref{bicriteria:sec} considers the most general setting; we introduce a notion of capacity augmentation, and present a polynomial-time algorithm 
%for the reallocation schedule of minimum completion time 
under augmented capacities. 
In Section \ref{nocap:sec}, we consider the setting where the capacities of warehouses are sufficiently large, and present constant-factor approximation algorithms for the setting. Section \ref{hardness:sec} 
%summarizes 
shows
hardness results in various settings. %Due to the space limitation, we omit most of the detailed proofs, which can be seen in Appendix.   

\section{Preliminaries}\label{preliminaries:sec}
We first define the reallocation problem.
%and then present a standard  integer linear system formulation representing the feasibility (with a bounded time interval) of the problem. 
%Based on this, we provide an integer linear programming formulation of the problem. 
%The formulation are used  in  Section  \ref{bicriteria:sec} to propose approximation algorithms for several scenarios. 
We also introduce demand graphs of the problem, which plays a key role in designing an efficient algorithm  for the reallocation problem with uniform product size and transit time. 

%\subsection{Reallocation problem}\label{problem:sec}
Let  $\P$ be  a set  of \emph{\prods} (or items), and let $\W$ be   a set of \emph{\houses}, where $m=|\P|$ and $n=|\W|$.
Each \pro $p \in \P$ 
has size  $\size(p) \in \mathbb{R}_+$, 
where $\mathbb{R}_+$ denotes the set of nonnegative reals.
It  is initially stored in 
a \emph{source} \house $\source(p) \in \W$,
 required to be reallocated to a \emph{sink} \house $\sink(p) \in \W$, 
  and
its reallocation  from $\source(p)$  to $\sink(p)$
 takes  \emph{transit time} $\transit(p)\,( >0)$. 
 Here we assume that $s(p)\not=t(p)$ for all $p \in P$. 
 %$\in \mathbb{Z}_+$, 
% where $\mathbb{Z}_+$ denotes the set of nonnegative integers.
 Namely, if it is sent from $s(p)$ at time $\theta$, it reaches $t(p)$ at time $\theta + \transit(p)$. 
Each \house $w \in \W$ has capacity $\capa(w) \in \mathbb{R}_+$, 
which represents the upper bound of the total size of \prods  stored in $w$ at any time. 
Moreover it has carry-out and carry-in capacities $d^+(w), d^-(w) \in \mathbb{R}_+$,  
which respectively represent the upper bounds  of  the total size of
 \prods allowed to be sent from and be received at $w$ at every time.

We consider reallocation schedules of given \prods in the {\it discrete-time model},
which means that for any product $p \in P$,  sending time $\theta_p$ and transit time $\transit(p)$ are  nonnegative and positive integers, respectively.  
For a \house $w \in \W$, let $\P^+(w)$ and $\P^-(w)$ respectively denote the sets of 
\prods $p \in \P$
initially and finally stored at $w$, i.e.,  
 $\P^+(w)=\{p \in \P \mid \source(p)=w\}$
and $\P^-(w)=\{p \in \P \mid \sink(p)=w\}$.
A {\em   reallocation schedule} is a mapping $\depart:P\to \mathbb{Z}_+$, where  $\mathbb{Z}_+$ denotes the set of nonnegative integers, 
and  is called {\em feasible}  
if it satisfies the following three conditions:

\begin{romanenumerate}
%\item Every \pro $p \in \P$ is reallocated from $\source(p)$
%5to $\sink(p)$.

\item At each time  $\tt \in \mathbb{Z}_+$, the total size of \prods departing from $w$
is at most $d^+(w)$ for every \house $w \in \W$, i.e., $\sum_{p \in \P^+(w): \depart(p)=\tt} \size(p) \leq d^+(w)$.
 
\item At each time  $\tt \in \mathbb{Z}_+$, the total size of \prods arriving at $w$
is at most $d^-(w)$ for every \house $w \in \W$, i.e.,
 $\sum_{p \in \P^-(w): \depart(p)=\tt-\transit(p)} \size(p) \leq d^-(w)$.

\item At each time  $\tt \in \mathbb{Z}_+$, the total size of \prods located in  $w$
is at most $\capa(w)$ for every \house $w \in \W$, i.e., $\sum_{p \in \P^+(w): \depart(p)\geq \tt} \size(p)+
\sum_{p \in \P^-(w): \depart(p) \leq \tt-\transit(p)} \size(p) \leq \capa(w)$.
\end{romanenumerate}

\noi
Constraints (i) and (ii) are respectively called   \emph{carry-out} and \emph{carry-in}
\emph{capacity
constraints}, and Constraint (iii)  is  called \emph{\house capacity constraint}.
In this paper, we assume that  
$\size(p) \leq d^+(w)$ for all $p \in \P^+(w)$, 
 $\size(p) \leq d^-(w)$ for all $p \in \P^-(w)$, 
  $\sum_{p \in \P^+(w)}  \size(p) \leq \capa(w)$, and 
   $\sum_{p \in \P^-(w)}  \size(p) \leq \capa(w)$, 
since they are clearly necessary conditions for the feasibility and can be checked in linear time.  
Our problem, called
the \emph{reallocation problem}, is to compute a feasible reallocation schedule with the minimum completion time. 
Here the completion time $T$ of a reallocation schedule $\depart$ is defined as $
T = \max_{p \in \P}\{\depart(p)+\transit(p)\}$. 

\ishii{
We here remark that if
both carry-in and carry-out capacities are
sufficiently large,
then the problem is trivial and
 an optimal reallocation schedule can be obtained by sending all \prods in time
 0, i.e., letting $\depart(p)=0$ for all \prods $p \in \P$.
 Indeed, in this case,
 the carry-in and carry-out capacity constraints are always %automatically 
 satisfied.
The \house capacity constraints are also
 satisfied since
 every \house $w \in \W$ has
 no \pro in $\P^+(w)$ in it
 at each time $\tt>0$ and satisfies
$\sum_{p \in \P^-(w)}  \size(p) \leq \capa(w)$ by
assumption.

\begin{theorem}\label{trivial:th}
The reallocation problem
can be solved in
    $\mO(m+n)$ time if $d^+=d^- \equiv \infty$.
\end{theorem}
}

Before ending this section, let us fix some notation on graphs and define demand graphs, which frequently appear in the subsequent sections.   
An undirected or directed graph  $G$ is an ordered pair of its vertex set $V(G)$ and edge set $E(G)$ 
and is denoted by $G=(V(G), E(G))$, or simply $G=(V, E)$.
In an undirected graph  $G$, the \emph{degree}  of a vertex $v$, 
denoted by 
$\delta_G(v)$, is the number of edges incident to $v$. 
In a directed graph $G$, the \emph{out-degree} and  \emph{in-degree}
of a vertex $v$ are respectively defined as 
$\delta^+_G(v)=|\{(v,w) \in E(G) \mid w\in V(G)\}|$ and $\delta^-_G(v)=|\{(u,v) \in E(G) \mid u\in V(G)\}|$.
%and the \emph{degree} of $v$, denoted by $\deg_D(v)$, 
%equals $\deg_D^+(v)+\deg_D^-(v)$. 
We denote by $\Delta^+(G)$ and  $\Delta^-(G)$ the maximum out-degree  and  in-degree of a digraph $G$, respectively. 
Given a set $\W$ of \houses and a set $\P$ of \prods, 
 we represent the demand relationship for \prods  as a directed multigraph   
 $G=(\W, \{(\source(p),\sink(p)) \mid p \in \P\})$, which is called  the \emph{demand graph} of the reallocation problem. 
%We call such a digraph a \emph{\dmgraph}.
%We often refer to an arc of a \dmgraph
% corresponding to a \pro $p \in \P$ as $a_p$. 
%For simplicity, we often use the same notation to refer to corresponding
%\prods in $\P$ and arcs in $D$.

\nop{
Let $D$ be a requirement graph.
For a directed path $P$ in  $D$, we consider a simultaneous
transportation of \prods corresponding arcs in $\P(P)$,
which means that we leave all \prods in $\P(P)$ at the same time.
We call such a transportation a {\em path-move $($with respect to a path
$P)$}.  
Particularly, it is called  
 a {\em cycle-move $($with respect to a path
$P)$} if $P$ is a directed cycle.
}

%\section{Tractable Cases}\label{poly:sec}

\section{Uniform \pro size and transit time}\label{size2:sec}

%%%%%%%%%%%%%%%%%%%%%%%%%%
In this section, we consider  the reallocation problem with the  basic scenario in which product size  and transit time  are both uniform.
We show that by using a cycle decomposition of the demand graph, the reallocation problem can be solved in polynomial time. 
More precisely, we have the following result. 

\begin{theorem}\label{poly:th}
The reallocation problem  with uniform product size and uniform transit time
%$\Q_{\max}+\lambda-1$ is the minimum completion time.  
%the reallocation problem
can be solved in
    $\mO(mn \log m)$ time.
\end{theorem}
Here we recall that $m=|P|$ and $n=|W|$. 
Before proving the theorem, let us provide a general lower bound on the minimum completion time. 
For a warehouse $w \in W$, let 
\begin{eqnarray*}
\Q(w) &=& \max\left\{\left\lceil \frac{\sum_{p \in \P^+(w)}\size(p)}{d^+(w)} \right\rceil, 
\left\lceil \frac{\sum_{p \in \P^-(w)}\size(p)}{d^-(w)} \right\rceil\right\}, \\
\Q_{\max}&=&\max_{w \in \W} \Q(w).
\end{eqnarray*}
Then we have the following lemma. 

\begin{lemma}\label{low:lem}
For the reallocation problem, the minimum completion time is at least 
$\Q_{\max} + \min_{p \in \P}\{ \transit(p)\}-1.$
\end{lemma}
\begin{proof}
By \outdegree  capacity constraints, for any warehouse $w \in W$, we need at least  
$\lceil \frac{\sum_{p \in
\P^+(w)}\size(p)}{d^+(w)} \rceil$ steps to send all the products $p \in P^+(w)$. Thus the minimum completion time is at least $\lceil \frac{\sum_{p \in
\P^+(w)}\size(p)}{d^+(w)} \rceil +\min_{p \in \P}\{ \transit(p)\} -1$. 
Similarly, by \indegree  capacity constraints, 
the minimum  completion time is at least $\lceil \frac{\sum_{p \in
\P^-(w)}\size(p)}{d^-(w)} \rceil +\min_{p \in \P}\{ \transit(p)\} -1$, which proves the lemma. 
\end{proof}

Note that the completion time might be far from the lower bound in Lemma \ref{low:lem}.
However, we below show that it  matches the lower bound $\Q_{\max}+\mu-1$ if product size and transit time are both uniform, where 
 $\mu$ denotes the uniform transit time, i.e.,  $\mu=\tau(p)$ for all $p \in P$.

%\noi
%Hence, 
%the case where both of $\size$ and $\transit$ are uniform
%is polynomially solvable;
%namely, we have the following theorem.

In the rest of  this section, 
we assume without loss of generality that 
\begin{equation}\label{s=1:eq}
 \size(p)  =  1
\end{equation}
 for all \prods $p \in \P$, 
and all functions $d^+$, $d^-$, and $\capa$ are integral.
This is because
an problem instance equivalent to the original one is obtained by using 
$\tilde{d}^+(w)= \lfloor \frac{d^+(w)}{k} \rfloor$, 
$\tilde{d}^-(w)= \lfloor \frac{d^-(w)}{k} \rfloor$, and
$\tilde{\capa}(w)=\lfloor \frac{\capa(w)}{k} \rfloor$ for all $w \in \W$, if $\size(p)=k$ for all $p\in \P$. 
Note that in this case we have 
$\Q(w)=\max\{\lceil \frac{|\P^+(w)|}{d^+(w)} \rceil, \lceil \frac{|\P^-(w)|}{d^-(w)} \rceil\} $ 
 for all $w \in W$. 
Thus we have $\Q_{\max} \leq m$.

Let us now present a simple but important observation of feasible reallocation schedules. 
For two positive integers $a$ and $b$ with $a \leq b$, let $[a,b]=\{a,a+1, \dots, b\}$.  
Let $Q=\{q_1,q_2,\ldots,q_\ell\} $ be
 a set  of \prods which forms  a simple cycle in the demand graph $G=(W,E)$, i.e., 
 $Q$ satisfies 
 $\sink(q_i)=\source(q_{i+1})$  for $i \in [1,\ell]$  and 
  $\source(q_i)\neq \source(q_j)$ for any distinct $i$ and $j$, where $q_{\ell+1}$ is defined as $q_{1}$. 
Then we claim that 
all products in $Q$ can be sent  simultaneously.

Consider a situation where all \prods in $Q$ depart at time $\tt$ (and no other product departs at time $\tt$).
By our assumption that $\size(p)\leq d^+(w)$ for 
 every \pro $p \in \P^+(w)$ and $\size(p)\leq d^-(w)$ for 
 every \pro $p \in \P^-(w)$
as mentioned in Section~\ref{preliminaries:sec}, the \outdegree and \indegree capacity
constraints
are satisfied.
Moreover, $\sink(q_i)$ has a room for $q_i$'s arrival at time $\tt+\mu$,  since $q_{i+1}$ leaves from $\source(q_{i+1})\,(=\sink(q_i))$ at time $\tt$. 
Thus  \house capacity constraints are also satisfied,
which implies the claim.
More generally, if a set $Q$ of products  can be partitioned into vertex-disjoint simple cycles, then they can be sent simultaneously.  

Based on this observation, we  construct an efficient algorithm for the reallocation problem when 
 product size  and transit time  are both uniform.
In order to explain it smoothly, we first consider the following subcase: 
\begin{eqnarray}
 &&  \  d^+(w)  =  d^-(w)  =  1 \mbox{ and }  |\P^+(w)|  =  |\P^-(w)| \mbox{ for  every \house  }\ 
w \in \W.  \label{eq-aa1}
\end{eqnarray}
We show that $P$ can be partitioned into $\Q_{\max}$ sets $P_i$ ($i \in [1,\Q_{\max}]$), each of which forms  vertex-disjoint simple cycles in the demand graph. 
This implies the existence of a feasible reallocation schedule with the completion time $\Q_{\max}+\transituni-1$. By Lemma \ref{low:lem}, we can see that it is an optimal schedule. 
Note that $\Q_{\max}=\max_{w \in W}|\P^+(w)|$ by (\ref{eq-aa1}).

Let $H$ be the  bipartite graph obtained from the demand graph $G=(W,E)$ 
by creating two copies $w_1$ and $w_2$ of each vertex $w \in \W$ 
and adding an edge $(w_1, v_2)$ for each $(w,v) \in E$; namely, $V(H)=W_1\cup W_2$ and
$E(H)=\{(w_1,v_2) \mid (w,v) \in E\}$,
where $W_i =\{w_i \mid w \in \W\}$ for $i=1,2$.
By  assumption (\ref{eq-aa1}), $\Q_{\max}$ represents
the maximum degree $\Delta(H)$ of $H$ and 
  $\delta_{H}(w_1)=\delta_{H}(w_2)$ holds for all $w \in \W$.
Let us further modify the graph $H$.
Let $H^*$ be the bipartite graph obtained from $H$ by adding 
 $\Q_{\max}-\delta^+_G(w)$ multiple edges $(w_1, w_2)$
for every $w \in \W$. 
Note that $H^*$ is $\Q_{\max}$-regular,
where a graph is called \emph{$d$-regular} if every vertex has degree  $d$. 
It is known that the edge set of a $d$-regular bipartite graph
can be partitioned into $d$ perfect matchings \cite{Konig16}. 
Moreover, we have the following lemma. 

\begin{lemma}
\label{lemma-000a}
Let $M^*$ be a perfect matching  in $H^*$.
Then $M^*$ corresponds to vertex-disjoint simple cycles in the demand graph $G$.  
\end{lemma}
\begin{proof}
%Let $V' \subseteq V$ be the set of vertices  in $D$ corresponding to
%the end-vertices of  edges in $E'$. 
For a perfect matching  $M^*$ in $H^*$,  let 
$M=M^* \cap E(H)$ and let $H[M]$ be the subgraph of $H$ induced by $M$. 
%Note that any edge $(w_1,w_2) \in E(H^*)$ added to $H$ can be regarded as self-loop in the demand graph 
Then we note that $M$ is a matching of $H$ such that 
the degree of $w_1$ with respect to $H[M]$
is equal to that of $w_2$ for any $w \in W$. 
Hence, $M^*$ corresponds to  vertex-disjoint simple cycles in the demand graph $G$.
\end{proof}

\noi
The next lemma follows from Lemma \ref{lemma-000a} and the argument before it.

\begin{lemma}\label{matching:lem}
If product set $P$ satisfies $|\P^+(w)|  =  |\P^-(w)|$  for  every \house $w \in \W$, 
then it  can be partitioned into  sets $P_i$ $(i \in [1,\max_{w \in W} |\P^+(w)|])$ such that each $P_i$ forms  vertex-disjoint simple cycles in the demand graph. 
\end{lemma}
\nop{
\begin{proof}
It follows from Lemma \ref{lemma-000a} and the discussion before it. 
\end{proof}
}
%Consequently, we can partition $\P$ into $\P_t$,
% $t \in [0,p-1]$
%such that every $\P_t$ is a family of  pairwise vertex-disjoint cycles in $D$
%by computing   $p$  perfect matching in $G_D'$.
Consequently, if (\ref{eq-aa1}) is satisfied, 
by setting $\depart(p)=i-1$ if $p \in P_i$, 
we obtain an optimal reallocation schedule $\depart$, 
whose completion time is equal to  $\Q_{\max}+\transituni-1$.

We next consider a slightly generalized case in which
for   every \house $w \in \W$, $d^+(w)  =  d^-(w)  =  1$ holds, but 
$|\P^+(w)|  =  |\P^-(w)|$  does not necessarily hold. 
Let $\W^+$ and $\W^-$ respectively denote  the sets of \houses $w$ with  
 $|\P^+(w)|>|\P^-(w)|$ and $|\P^+(w)|<|\P^-(w)|$. 
 Then we add extra \prods $p$ such that $\source(p)\in \W^-$, $\sink(p) 
\in \W^+$, and $\size(p)=1$, 
until $|\P^+(w)|  =  |\P^-(w)|$  holds for all $w \in \W$.
Note that 
it can be done by arbitrarily pairing \houses in $\W^-$ and $\W^+$, since $\sum_{w\in \W}|\P^+(w)|=\sum_{w\in \W}|\P^-(w)|$.
We also claim that every \house has enough vacancy for such extra \prods.
Indeed, let $\tilde{P}$ denote the resulting product set. Then we have $|\tilde{P}^+(w)|=|\tilde{P}^-(w)| \leq  \max\{|\P^+(w)|,|\P^-(w)|\}$, which implies that  \house capacity constraints at the initial time and the last time are satisfied if all products are sent correctly. 
Moreover, we can see that the schedule of $\tilde{P}$  for the previous case provides a schedule of $P$ that satisfies \house capacity constraints.

We finally consider the  general case.
Here we show that a schedule with the completion time
$\Q_{\max}+\transituni-1$ can be obtained
by reducing it to the case in which $d^+ =d^- \equiv   1$.
For a warehouse $w \in W$, let $h_w^+=\lceil \frac{|\P^+(w)|}{\Q_{\max}} \rceil$,
$h_w^-=\lceil \frac{|\P^-(w)|}{\Q_{\max}} \rceil$, and
$h_w =\max\{h_w^+,h_w^-\}$. 
By definition, we have  $h_w^+ \leq d^+(w)$ and $h_w^- \leq d^-(w)$. 
For  each \house $w \in \W$, we construct 
$h_w$ many  \houses $w_i$ ($i \in [1, h_w])$ 
such that $\tilde{d}^+(w_i)=\tilde{d}^-(w_i)= 1$, 
$\capa(w_i)=\Q_{\max}$ if $i \in [1,h_w-1]$, and $\capa(w_{h_w})=\capa(w)
-\Q_{\max}(h_w-1)$. 
Let $\tilde{W}$ denote the resulting set of warehouses. 
We then modify  source and sink \houses $\source$ and 
$\sink$ %by $\tilde{\source}$ and $\tilde{\sink}$
in such a way that
\begin{itemize}
\item For every produce $p$ with $\source(p)=w$, let $\tilde{\source}(p)=w_i$ for some $i \in [1,h_w^+]$,  
\item For every product $p$ with $\sink(p)= w$,  let $\tilde{\sink}(p)=w_i$ for some $i \in [1,h_w^-]$,

\item $|\tilde{\P}^+(w_i)|= \Q_{\max}$ for  every $i \in [1,h_w^+ -1]$ and 
 $|\tilde{\P}^+(w_{h_w^+})|=|\P^+(w)|-\Q_{\max} (h_w^+-1)$, 

\item 
 $|\tilde{\P}^-(w_i)|= \Q_{\max}$ for  every $i \in [1,h_w^- -1]$ and 
      $|\tilde{\P}^-(w_{h_w^-})|=|\P^-(w)|-\Q_{\max} (h_w^--1)$.
\end{itemize}
\noi
Here $\tilde{\source}$ and $\tilde{\sink}$ respectively represent the resulting source and sink warehouses, $\tilde{P}^+(x)=\{ p \in P \mid \tilde{\source}(p)=x\}$, and  $\tilde{P}^-(x)=\{ p \in P \mid \tilde{\sink}(p)=x\}$. 
Note that this can be done by numbering products in $\P^+(w)$
(resp., $P^-(w)$)
as $p_j$, $j \in [1,|\P^+(w)|]$ (resp., $[1,|\P^-(w)|]$), and letting $\tilde{\source}(p_j)=
w_{\lceil \frac{j}{\Q_{\max}} \rceil}$
(resp.,  $\tilde{\sink}(p_j)=w_{\lceil \frac{j}{\Q_{\max}} \rceil}$).
%Let $\tilde{\P}$ be the resulting set of  \prods.
Then by  applying
the discussion in the above case, we can obtain a schedule $\depart$  with
the completion time $\Q_{\max}+\transituni-1$, since $|\tilde{P}^+(w)|, |\tilde{P}^-(w)|\leq \Q_{\max}$
holds for all $w \in {\tilde{\W}}$. 
It is not difficult to see that this reallocation schedule $\depart$ is also feasible with the original problem instance.
By Lemma \ref{low:lem}, it is an optimal reallocation schedule. 
We describe the whole procedure mentioned above as 
Algorithm \alg{Uniform}.
%($\W,\P,d^+,d^-,\capa,\size,\transit$).
%Algorithm~\ref{matching:algo}.

We then analyze the time complexity of the scheduling algorithm.
Let $\tilde{\P}_1$ be the set of \prods obtained from $\tilde{\P}$ 
by adding extra \prods according to the above second case.
Note that $|\tilde{\P}_1| \leq 2m$.
Let $\tilde{G}$ be the \dmgraph for the  \mainpro with $\tilde{\W}$ and $\tilde{\P}_1$.
Note that we have $\delta^+_{\tilde{G}}(w)=\delta^-_{\tilde{G}}(w)$
for all $w \in \tilde{\W}$ and
$\Delta^+(\tilde{G})=\Delta^-(\tilde{G})= \Q_{\max}$.
Let $\tilde{H}$ and $\tilde{H}^*$ be the undirected bipartite 
graph 
and the undirected $\Q_{\max}$-regular bipartite graph
obtained from
$\tilde{G}$ according to the discussion immediately before  Lemma~\ref{lemma-000a},
respectively.
Note that $|E(\tilde{H})|=|\tilde{\P}_1| \leq 2m$ and
that $|E(\tilde{H}^*)|-|E(\tilde{H})| \leq \Q_{\max}n$
since the possible $i\in [1,h_w]$ with $\delta^+_{\tilde{G}}(w_i)<\Q_{\max}$ is only $h_w$
for each $w \in \W$.
%Note that $H'_{G_2}$ is a $\Q$-regular bipartite graph.
A family of $\Q_{\max}$ perfect matchings in $\tilde{H}^*$ can be obtained by
computing a minimum edge-coloring of $\tilde{H}^*$
in $\mO(|E(\tilde{H}^*)|\log{\Q_{\max}})=
\mO((m+n\Q_{\max})\log{\Q_{\max}})$ time \cite{COS01}.
Consequently,  we have Theorem~\ref{poly:th}.

\begin{algorithm}
\caption{{\alg{Uniform}($\W,\P,d^+,d^-,\capa,\size,\transit$)}}
\begin{algorithmic}[1]
\renewcommand{\algorithmicrequire}{\textbf{Input:}}
 \renewcommand{\algorithmicensure}{\textbf{Output:}}
\label{matching:algo}
\REQUIRE An instance of the \mainpro
with uniform 
 $\size$ and uniform $\transit$.
%where $\transit \equiv \ell$.

% A set $X$ of variables of (\ref{IP:eq})

\ENSURE A schedule for $\P$  with the completion time
$\Q_{\max}+\transituni-1$,
where  
$\Q_{\max}=\max_{w \in \W}\{\max\{\lceil \frac{\sum_{p \in \P^+(w)}\size(p)}{d^+(w)} \rceil, $ 
$
\lceil \frac{\sum_{p \in \P^-(w)}\size(p)}{d^-(w)} \rceil\}\}$.

\FOR{$w\in \W$}

\STATE Divide $w$ into $d_w$ \houses $w_i$, $i \in [1, d_w]$.

\STATE Number \prods in $\P^+(w)$ as $p_j$, $j \in [1,|\P^+(w)|]$,
and let  $\source(p_j)=w_{\lceil \frac{j}{\Q_{\max}} \rceil}$ 
%with $i = \lceil \frac{j}{p} \rceil$ 
for every
 $p_j \in \P^+(w)$.

\STATE Number \prods in $\P^-(w)$ as $p_j$, $j \in [1,|\P^-(w)|]$
and let  $\sink(p_j)=w_{\lceil \frac{j}{\Q_{\max}} \rceil}$ 
%with $i = \lceil \frac{j}{p} \rceil$
 for every
 $p_j \in \P^-(w)$.

\ENDFOR

\COMMENT{Denote the resulting set of \houses and \prods by $\tilde{\W}$ and
 $\tilde{\P}$,
respectively.}

\WHILE{$\exists w \in \tilde{\W}$ with $|\tilde{\P}^+(w)| \neq |\tilde{\P}^-(w)|$}

\STATE Add to $\tilde{\P}$ an extra \pro $p$ with $\source(p)=u$ and 
$\sink(p)=v$
for some $u,v \in \tilde{\W}$ with $|\tilde{\P}^+(u)| < |\tilde{\P}^-(u)|$ and
$|\tilde{\P}^+(v)| > |\tilde{\P}^-(v)|$.

\ENDWHILE

\STATE Compute a partition  $\{\P_i \mid i \in [1,\Q_{\max}]\}$   of
$\P$ such that 
in the corresponding \dmgraph $G=(\tilde{\W},E_{\tilde{\P}})$,
the set of directed edges in $E_{\tilde{\P}}$
 corresponding to $\P_i$
induces a family of 
 vertex-disjoint
 simple cycles in $G$.

\STATE Let $\depart(p)=i-1$ for all $p \in \P_i$, $i \in [1,\Q_{\max}]$, and output $\{\depart(p) \mid p \in \P\}$. 
\end{algorithmic}
\end{algorithm}

%%%%%%%%%%%%%%%%%%%%

\nop{
Also, we can observe that  
%Algorithm~\ref{matching:algo}
the above algorithm
%Algorithm \alg{Uniform}
delivers  a 2-approximate schedule for 
the case in which 
product size
%$\size$ 
is uniform
and
$d^-$ is sufficiently large, where we note  that
$\transit$ and $d^+$ are both general.
}

We finally remark that  
%Algorithm~\ref{matching:algo}
\alg{Uniform}
delivers  a 2-approximate schedule for 
the case in which 
product size
%$\size$ 
is uniform
and
$d^-$ is sufficiently large, even if
$\transit$ and $d^+$ are both general.
Let $\depart$ be the schedule obtained by
%Algorithm~\ref{matching:algo}.
\alg{Uniform}.
Then,  $\depart$ satisfies the \house capacity constraints, since
it is based on a cycle decomposition.
Also, by construction and $d^-\equiv +\infty$, the carry-out and carry-in capacity constraints
are satisfied.
This means that $\depart$ is feasible.
Observe that the completion time of $\depart$ is 
at most $\Q_{\max}+\max_{p  \in \P}\transit(p)$.
Since $\Q_{\max}$ and $\max_{p  \in \P}\transit(p)$ are both lower
bounds on the minimum completion time, it follows that 
$\depart$ is a 2-approximate schedule.
Therefore, we have the following theorem, where
the case of $d^+\equiv +\infty$ can be treated similarly. 

\begin{theorem}\label{2apx:th}
The reallocation problem  with uniform product size and $d^- \equiv
 +\infty$ $($or $d^+ \equiv  +\infty)$
%$\Q_{\max}+\lambda-1$ is the minimum completion time.  
%the reallocation problem
is 2-approximable  in
    $\mO(mn \log m)$ time. 
\end{theorem}

\section{Approximation Algorithms for General Cases}\label{bicriteria:sec}

As shown later in Theorem~\ref{hard1:th},
it is NP-complete to decide whether a reallocation schedule is  feasible.
Hence, we need to  replace some of the
hard constraints with soft constraints.
In this paper, we consider \emph{capacity augmentations}. Namely,
we  relax three capacity constraints: \outdegree, \indegree, and
\house capacity constraints.
For a warehouse $w \in W$, let $\sigma^+_{1}(w)$ and  $\sigma^+_{2}(w)$ 
respectively denote the largest and the second largest size of \prods in $\P^+(w)$, 
and similarly $\sigma^-_{1}(w)$ and  $\sigma^-_{2}(w)$ 
respectively denote the largest and the second largest size of \prods in $\P^-(w)$.

Then we obtain the following result. 

%and construct a reallocation schedule $\depart$ such that (1) it is feasible with such relaxed constraints and (2) the completion time of $\depart$ is at most the minimum completion time of the original problem 

\begin{theorem}\label{bicriteria:th}
{\rm (i)}  We can compute in polynomial time a schedule for the reallocation problem such that the completion time is at most the minimum completion time and it is feasible with the capacity augmentation in which  for every \house $w \in \W$, 

\begin{alphaenumerate}
\item the carry-out capacity $d^+(w)$ is augmented by  $\sigma^+_{1}(w)$ and  $\sigma^+_{2}(w)$, 

\item the carry-in capacity $d^-(w)$ is augmented by  $\sigma^-_{1}(w)$ and  $\sigma^-_{2}(w)$, and  

\item the warehouse capacity $c(w)$ is augmented by $c(w)$. 
\end{alphaenumerate}

\noi
{\rm (ii)} If the carry-in $($resp., carry-out$)$ capacity  is sufficiently large, i.e., $d^-\equiv \infty$
$($resp., $d^+\equiv \infty)$,
we can compute in polynomial time a schedule for the reallocation problem such that the completion time is at most the minimum completion time and it is feasible with the capacity augmentation of  $($c$)$ and $($d$)$  for every \house $w \in \W$, where  

\begin{romanenumerate}
\item[{\bf \textsf{\textcolor{lipicsGray}{(d)}}}] the \outdegree $($resp., carry-in$)$ capacity is augmented by $\sigma^+_{1}(w)$
$($resp., $\sigma^-_{1}(w))$.
\end{romanenumerate}
 \end{theorem}

\noi
Here, for example, the statement (a) denotes that 
$\sum_{p \in \P^+(w): \depart(p)=\tt} \size(p) \leq d^+(w)+ 
\sigma^+_1(w)+\sigma^+_2(w)$ for each time $\tt$. 

By the hardness result in Theorem~\ref{hard1:th},
we cannot obtain a feasible schedule for the reallocation problem, unless P=NP. 
However, by Theorem \ref{bicriteria:th},  
if we take capacity augmentation appropriately, we can compute a feasible schedule with the augmentation whose  objective value is not worse than the optimal value of the original problem. 

We  remark that the statement (c) always holds 
by the assumption in Section~\ref{preliminaries:sec} that the total size of \prods in $\P^+(w)$ 
(resp., $\P^-(w)$) is at most
$\capa(w)$. 
Hence, we do not consider the \house 
capacity constraints in the subsequent discussion of this section. 

We first show the general case (i).
%We then consider a problem for asking a feasible schedule with the completion
%time $\T$ for a given integer $\T$.
%%%%%%%%%%%%%%%%%%%%%%%%%%%%%%%%%%%%%%
For  a given integer $\T$ as an upper bound of the completion
time, 
let us represent
%the problem can be represented as 
an integer linear system formulation of the feasibility with $T$ of the reallocation problem.   
\begin{align}
%\label{IP0:eq}
% \mbox{minimize} & \sum_{t=0}^{\Q-1} \sum_{p  \in \P} c_{pt}x_{pt}& \\[.08cm] 
%\mbox{subject to}
&\sum_{p \in \P^+(w)} \size(p)x_{p\tt} \ \leq\  d^+(w)
&&  \forall  w\in \W, \ \forall \tt \in [0,\T-1] 
 \label{eq-a1}\\[.05cm]
& \sum_{p \in \P^-(w)} \size(p)x_{p,\tt-\transit(p)} \leq d^-(w)
& &\forall w\in \W,  \ \forall \tt \in [1,\T]
 \label{eq-a2}\\[.05cm]
%& \sum_{p \in \P^+(w)} \sum_{\lambda \geq \tt} \size(p)x_{p\lambda} &&\nonumber\\
%& \hspace*{1cm}+\sum_{p \in \P^-(w)} \sum_{\lambda \leq \tt}\size(p) x_{p,\lambda-\transit(p)}
%\ \leq \ \capa(w)
%&& \forall w\in \W,  \ \forall \tt \in [0,\T]
% \label{eq-a3}\\[.05cm]
& \sum_{\tt=0}^{\T-\transit(p)} x_{p\tt} \ = \ 1
&&\forall p\in \P
 \label{eq-a4}\\[.05cm]
& \hspace*{.5cm} x_{p\tt} \ \in \ \{0,1\}
&& \forall p\in \P, \  \forall \tt \in [0,\T-\transit(p)]\label{eq-a5}
\end{align}
Here  $x_{p\tt}$ denotes the indicator variable for a
\pro $p \in \P$ and departure
 time $\tt \in [0,\T-1]$; namely, $x_{p\tt}$ takes 1 if
a \pro $p$ departs at time $\tt$, and 0 otherwise.
Note that $x_{p\tt}$ with $\tt > \T-\transit(p)$ is not defined, since $p$
needs to arrive at $\sink(p)$ by time $\T$.
Inequalities \eqref{eq-a1} and \eqref{eq-a2} 
%and \eqref{eq-a3}
respectively correspond to carry-out and carry-in
%, and \house capacity
constraints.
Equality \eqref{eq-a4} ensures that every \pro $p \in \P$ is sent  exactly once by time $\T-\transit(p)$.
Note that the minimum $\T$ for which the integer linear system is feasible is the optimal completion time for the reallocation problem.  
%Thus if we know in advance an upper bound $T$ for the minimum completion time, then the reallocation problem 
%is to minimize $\sum_{p \in \P}\sum_{\lambda \leq T}  M^\lambda x_{p,\lambda-\transit(p)}$ so that it satisfies \eqref{eq-a1}--\eqref{eq-a5}, 
%where $M$ denotes a positive integer with $M\geq m+1$. 
%In  Section  \ref{bicriteria:sec}, we  show that the minimum completion time is polynomially bounded, and provide approximation algorithms for the reallocation problem by making use of linear relaxazation of the   integer linear system formulation. 

This formulation \eqref{eq-a1}--\eqref{eq-a5} is regarded as 
a problem of finding a feasible solution of
the \emph{2-sided placement problem} \cite{KMRT15}, which  is a generalization of 
the \emph{generalized assignment 
problem}.

Let $N$ denote a set of jobs, and let $M_1$ and $M_2$ be two disjoint sets of machines, where 
each job $k$ in $N$ is assigned to two machines $i \in M_1$ and $j \in M_2$. 
Let $F$ denote the set of possible assignments $F=\{(i,j,k) \in M_1 \times M_2 \times N \mid k \mbox{ can be assigned to $i$ and $j$}\}$.
Each machine $\ell \in M_1 \cup M_2$
has the resource capacity $d(\ell)$, and 
for an assignment $(i,j,k) \in F$,  the amount $s_1(i,j,k)$ (resp., $s_2(i,j,k)$) of resources of a machine $i$ (resp., $j$) is used if a job $k$ is assigned to machines $i \in M_1$ and $j \in M_2$. 
The 2-sided placement problem can be formulated as 
\eqref{eq-b0}--\eqref{eq-b4}:
\begin{align}
%\label{IP0:eq}
 \mbox{minimize} & 
\sum_{(i,j,k) \in F}
 c_{ijk}x_{ijk}
& \label{eq-b0}\\[.08cm] 
\mbox{subject to} &
\sum_{j, k: (i,j,k) \in F} s_1(i,j,k) x_{ijk} \ \leq\ d(i)
& &  \forall  i\in M_1 
 \label{eq-b1}\\[.05cm]
& \sum_{i, k: (i,j,k) \in F} s_2(i,j,k) x_{ijk} \ \leq\ d(j)
& & \forall  j\in M_2
 \label{eq-b2}\\[.05cm]
%& \sum_{p \in \P^+(w)} \sum_{\lambda \geq \tt} \size(p)x_{p\lambda} &&\nonumber\\
%& \hspace*{1cm}+\sum_{p \in \P^-(w)} \sum_{\lambda \leq \tt}\size(p) x_{p,\lambda-\transit(p)}
%\ \leq \ \capa(w)
%&& \forall w\in \W,  \ \forall \tt \in [0,\T]
% \label{eq-a3}\\[.05cm]
& 
\sum_{i,j: (i,j,k) \in F} x_{ijk} \ = \ 1
&& \forall k \in N
 \label{eq-b3}\\[.05cm]
& \hspace*{.5cm}  x_{ijk} \in \{0,1\}
&& \forall  (i,j,k) \in F.
\label{eq-b4}
\end{align}
\noi
Here $x_{ijk}$ and $c_{ijk}$ respectively denote the indicator variable and the
assignment cost  for an assignment
 $(i,j,k) \in F$.
Constraints \eqref{eq-b1} and \eqref{eq-b2} ensure that
the total amount of resources needed for the assignment to a machine $\ell$
is at most the resource capacity $d(\ell)$.

The feasibility of the reallocation problem formulated by  
%$(\ref{IP:eq})$
\eqref{eq-a1}--\eqref{eq-a5}
can be transformed into the one of the $2$-sided placement problem 
%$(\ref{2side:eq})$
\eqref{eq-b1}--\eqref{eq-b4}
as follows. 

\begin{itemize}
\item every \house $w \in \W$ at every time $\tt$ corresponds to both types
 of machines, denoted by $m_1(w,\tt) \in M_1$ and $m_2(w,\tt) \in M_2$, and
 the resource capacities of  $m_1(w,\tt)$ and   $m_2(w,\tt)$ are respectively defined  as $d^+(w)$ and $d^-(w)$. 

\item  every \pro $p \in \P$ corresponds to a job.  
Define $F=\{(i,j,p) \mid  i =m_1(\source(p),\tt),  
j=m_2(\sink(p),\tt+\transit(p)) \mbox{ for } p \in P, \tt \in [0,T-\transit(p)]\}$. 
For $(i,j,p) \in F$,  
let 
 $s_1(i,j,p) =s_2(i,j,p) =\size(p)$.  
 \end{itemize}

Korupolu et al.~\cite{KMRT15} proposed a polynomial-time  algorithm,
based on  an iterative approximation method, 
for the capacity augmentation of the $2$-sided placement problem.
For $i \in M_1$, let 
$s^{\max}_{i}=\max_{j, k: (i,j,k) \in F}s_1(i,j,k)$, and  for $j \in M_2$, let 
$s^{\max}_{j}=\max_{i, k: (i,j,k) \in F}s_2(i,j,k)$.

\begin{theorem}[\cite{KMRT15}]
\label{KMRT:thm}
For the $2$-sided placement problem 
%$(\ref{2side:eq})$
\eqref{eq-b0}--\eqref{eq-b4}, there exists a polynomial-time
algorithm for finding an assignment of jobs in $N$ to
machines in $M_1 \cup M_2$ whose cost is at most the optimal if 
 the resource capacities for $i \in M_1$ and  for  $j\in M_2$ 
 are respectively augmented with $2s_{i}^{\max}$ and $2s_{j}^{\max}$.
 \end{theorem}
%Algorithm~\ref{iterative:algo} 
Algorithm \alg{Iterative}($\W,\P,d^+,d^-, \size,\transit,\T$)
describes 
the algorithm of  Korupolu et al. for the formulation 
\eqref{eq-a1}--\eqref{eq-a5}. 
%(\ref{IP:eq})
%(after defining an arbitrary linear objective  function)

%
\begin{algorithm}
\caption{{\alg{Iterative}($\W,\P,d^+,d^-,
%\capa,
\size,\transit,\T$)}}
\begin{algorithmic}[1]
\renewcommand{\algorithmicrequire}{\textbf{Input:}}
 \renewcommand{\algorithmicensure}{\textbf{Output:}}
\label{iterative:algo}
\REQUIRE A linear relaxation  $\LP$
%$(\ref{IP:eq})_{LP}$
of formulation
 \eqref{eq-a1}--\eqref{eq-a5} obtained by relaxing every variable $x_{p\tt} \in \{0,1\}$
to $x_{p\tt}\geq 0$, where
%$(\ref{IP:eq})_{LP}$
$\LP$
is feasible and
$X$ denotes the set of all variables.

% A set $X$ of variables of (\ref{IP:eq})

\ENSURE A schedule for $\P$ satisfying (i) and (ii) of Lemma~\ref{additive:lem} for all $w \in \W$ and all $\tt \in [0,\T]$.

\STATE  $X' \leftarrow X$.

\WHILE{$X' \neq \emptyset$}

\STATE Find an extreme point
% $\bm{x}$ 
of  the current
$\LP$.
% $(\ref{IP:eq})_{LP}$.

\IF {$\exists x_{p\tt} \in \{0,1\}$}
\STATE Fix the value of $x_{p\tt}$ to the current value and $X' \leftarrow
X' \sm \{x_{p\tt}\} $.

\ENDIF

\IF {the current LP contains a  constraint in \eqref{eq-a1} with 
$\sum_{p \in \P^+(w)} (1-x_{p\tt}) \leq 2$
 or a constraint in \eqref{eq-a2} with 
 $\sum_{p \in \P^-(w)} (1-x_{p,\tt-\transit(p)}) \leq 2$}
\STATE Remove the corresponding constraint from the current 
$\LP$.
%$(\ref{IP:eq})_{LP}$.

\ENDIF
\ENDWHILE

\STATE Let $\depart(p)=\tt$ if $x_{p\tt}=1$ for all $p\in \P$
and all $\tt \in [0,\T]$, output
$\{\depart(p) \mid p \in \P\}$, and halt.

\end{algorithmic}
\end{algorithm}
%
%
%and $c({\cal S}^*(\Q))$ be its objective function value.
%Since problem $(\ref{IP:eq})_{LP}$ is a relaxation of (\ref{IP:eq}) and
%the algorithm proceeds by relaxing some constraints,
%we can observe that $c({\cal S}^*(\Q))$
%is at most   the optimal value of problem (\ref{IP:eq}).
As shown in \cite{KMRT15},
%Algorithm~\ref{iterative:algo} 
\alg{Iterative}($\W,\P,d^+,d^-, \size,\transit,\T$)
works whenever
the  linear relaxation of a given formulation 
%(\ref{IP:eq})
\eqref{eq-a1}--\eqref{eq-a5}
is feasible.

For a schedule $\varphi$
obtained by 
%Algorithm~\ref{iterative:algo} 
\alg{Iterative}
%($\W,\P,d^+,d^-, \size,\transit,\T$)
with $\T$,
let $\P^+(w,\tt) \subseteq \P$ (resp., $\P^-(w,\tt)$)
 be the set of \prods departing from (resp., arriving at) a \house 
$w \in \W$ at time $\tt \in [0,\T]$.
Let   $p^+_1(w,\tt)$ and $p^+_2(w,\tt)$  respectively denote the  products in $\P^+(w,\tt)$
 with
the largest and second largest size, and let 
 $p^-_1(w,\tt)$ and $p^-_2(w,\tt)$  respectively denote the products in $\P^-(w,\tt)$
 with
the largest and second largest size.

The following lemma gives an upper bound for the capacity augmentation. 
%additive errors of constraints,
Note that it slightly improves 
the statement of Theorem \ref{KMRT:thm}, but it is necessary to obtain the results in the next sections.  
\begin{lemma}\label{additive:lem}
Algorithm 
%\ref{iterative:algo} 
\alg{Iterative}$(\W,\P,d^+,d^-, \size,\transit,\T)$
outputs a schedule that satisfies the following two conditions. 
\begin{romanenumerate}
 \item  $\sum_{p \in \P^+(w,\tt)} \size(p) \leq d^+(w)+ \size(p^+_1(w,\tt))
+\size(p^+_2(w,\tt))$. 

\item $\sum_{p \in \P^-(w,\tt)} \size(p) \leq d^-(w)+ \size(p^-_1(w,\tt))
+\size(p^-_2(w,\tt))$. 
\end{romanenumerate}
\end{lemma}
\begin{proof}
We only prove (i), since 
(ii) can be shown similarly.
%It suffices to show that 
% $\sum_{p \in \P^+(w,t)} \size(p) \leq d^+(w)+ \size(p_1^*)+
%\size(p_2^*)$ by the maximality of $\size(p_1^*)$ and $\size(p_2^*)$,
% where
%let $p_1^*$ (resp., $p_2^*$) be a \pro in $\P^+(w)$ with the maximum size
%(resp., the second maximum size).
We assume that  the corresponding \outdegree capacity constraint is removed during the $\gamma$th iteration  of the while-loop in
% Algorithm~\ref{iterative:algo}. 
\alg{Iterative}($\W,\P,d^+,d^-, \size,\transit,\T$).
Note that such an iteration  must exist, since otherwise  the \outdegree capacity constraint $\sum_{p \in \P^+(w,\tt)} \size(p) \leq d^+(w)$ is satisfied, which implies (i). 
Let $x^*_{p\tt}$'s denote the values of the LP  computed in the $\gamma$th iteration. 
For $i=0,1$, let $Q_i$ be the set of \prods  $p \in \P^+(w)$
such that the value of  $x_{p\tt}$ has been fixed to $i$ by the $\gamma$th iteration, 
and
$d_1=\sum_{p \in Q_1} \size(p)$.  
Note that $d_1+ \sum_{p \in \P^+(w) \sm (Q_0 \cup Q_1)} \size(p)x^*_{p\tt} \leq  d^+(w)$ and 
 $ \P^+(w,\tt)\subseteq \P^+(w) \sm Q_0$. 
By 
%$\P^+(w,\tt) \subseteq \P^+(w)$, 
$\P^+(w,\tt) \subseteq \P^+(w) \sm Q_0$, 
we have
\begin{equation}\label{remove:eq}
d_1+ \sum_{p \in \P^+(w,\tt) \sm  Q_1} \size(p)x^*_{p\tt} \leq  d^+(w).  
\end{equation}
Since the constraint is removed in the $\gamma$th iteration, we also have
$\sum_{p \in \P^+(w) \sm  (Q_0 \cup Q_1)} (1-x^*_{p\tt})\leq 2$, which again by   $\P^+(w,\tt) \subseteq \P^+(w)\sm Q_0$ implies 
\begin{equation}\label{additive:eq}
 \sum_{p \in \P^+(w,\tt) \sm   Q_1} (1-x^*_{p\tt})\leq 2.
\end{equation}
%(note that during the algorithm, 
% every variable $x_{p\tt}$ satisfies
% $0 \leq x_{p\tt} \leq 1$ 
%by the constraints
%\eqref{eq-a4}).
%in (\ref{IP:eq})).

Observe that
\begin{eqnarray}
\nonumber
\sum_{p \in \P^+(w,\tt)} \size(p) & = &
 d_1+\sum_{p \in \P^+(w,\tt) \sm    Q_1} \size(p) \\
\nonumber
& = & d_1 + \!\!\!\! \sum_{p \in \P^+(w,\tt) \sm   Q_1} \size(p)x^*_{p\tt} 
+ \!\!\!\!\sum_{p \in \P^+(w,\tt) \sm  Q_1} \size(p)(1-x^*_{p\tt})\\
%\nonumber
& \leq & d^+(w) +\sum_{p \in \P^+(w,\tt) \sm   Q_1} \size(p)(1-x^*_{p\tt}), \label{eq--xx0}
\end{eqnarray}
where the last inequality follows from (\ref{remove:eq}).
Let $p_1$ and $p_2$ denote the products in 
$ \P^+(w,\tt) \sm Q_1$ with the largest and second largest sizes, respectively. 
Then we have 
\begin{eqnarray}
\nonumber
\sum_{p \in \P^+(w,\tt) \sm  Q_1} \size(p)(1-x^*_{p\tt}) & = & 
(\size(p_1)- \size(p_2))(1-x^*_{p_1 \tt}) + \size(p_2)(1-x^*_{p_1 \tt})
\\
\nonumber 
&&
+\sum_{p \in \P^+(w,\tt) \sm (Q_1 \cup \{p_1\})} \size(p)(1-x^*_{p\tt})
\\
\nonumber
&\leq & 
 (\size(p_1)- \size(p_2))
 \\
 \nonumber
 &&
 + \size(p_2) 
\sum_{p \in \P^+(w,\tt) \sm  Q_1} (1-x^*_{p\tt})
\\
\nonumber
& \leq &  \size(p_1)+ \size(p_2)  \\
& \leq & \size(p^+_1(w,\tt))+ \size(p^+_2(w,\tt)).\label{eq--xxa}
\end{eqnarray}
Here the second inequality follows from
 (\ref{additive:eq}), and 
 the first and third ones follow from
 the definitions of $p_i$ and $\size(p^+_i(w,\tt))$, respectively.
By (\ref{eq--xx0}) and (\ref{eq--xxa}), we obtain the property (i).
\end{proof}

Let $\T_{\min}$ be the minimum $\T$ for which
the linear relaxation of 
formulation
\eqref{eq-a1}--\eqref{eq-a5}
%problem (\ref{IP:eq})
is feasible, and let $\varphi$ be a schedule obtained by %Algorithm \ref{iterative:algo}.
\alg{Iterative}
with $T=T_{\min}$.
%($\W,\P,d^+,d^-, \size,\transit,\T$).
It is easy to see that $T_{\min}$ is the completion time of $\varphi$, which is at most    
the minimum completion time of the original reallocation problem. 
Moreover, Lemma~\ref{additive:lem} 
implies that $\varphi$ is feasible with the capacity augmentation mentioned in Theorem \ref{bicriteria:th}. 

We thus remain to show that (I) $T_{\min}$ can be computed in polynomial time and (II) the size of linear relaxation with $T=T_{\min}$ in 
Algorithm \alg{Iterative}
%Algorithm \ref{iterative:algo}
is bounded by a polynomial in the input size.
The following lemma proves both (I) and (II). 
%and 
%by using binary search for example, we obtain (I), i.e., $T_\min$ can be computed in polynomial time. 
%%%%%%%%%%%%%%%%%%%%%%%%%%%
We therefore obtain Theorem \ref{bicriteria:th} (i).  

\begin{lemma}\label{feasible1:lem}
A linear relaxation  
of formulation
 \eqref{eq-a1}--\eqref{eq-a5} with $T=T_{\min}$
has a feasible solution $x^*_{p\tt}$ \,$(p\in P, \tt\in [0,T_{\min}-\tau(p)])$ such that 
 $x^*_{p\tt}=0$ for all $p \in P$ and all $\tt \geq 2m$. 
\end{lemma}
\begin{proof}
Suppose that $x^*_{q\eta} >0$ holds for some $q \in P$ and $\eta \geq 2m$. 
Note that we have $\sum_{p \in P:s(p)=s(q)} \sum_{\tt\leq 2m-1}x^*_{p\tt}$
$\leq m-x^*_{q\eta}$, i.e., 
 at most $m-x^*_{q\eta}$ many products are sent from $s(q)$ during the time interval $[0,2m-1]$. 
Thus warehouse $s(q)$ has room to send at least
\begin{eqnarray*}
%\Theta&=&\{\theta \in \mathbb{Z}_+ \mid \tt \leq 2m-2, \sum_{q \in P:s(q)=s(p)} x^*_{q\tt} < d^+(s(p))\}\\
f&=&\sum_{\tt \in [0,2m-1]}\bigl(1 -\sum_{p \in P:s(p)=s(q)}x^*_{p\tt}\bigr) \geq 2m-(m-x^*_{q\eta})=m+x^*_{q\eta}
\end{eqnarray*}
 many products during the time interval $[0,2m-1]$, where we note that at least one product can be sent at any time.  
 On the other hand, since 
\begin{eqnarray*}
\sum_{\tt \in [0,2m-1]}\sum_{p \in P:t(p)=t(q)}x^*_{p(\tt+\tau(q)-\tau(p))}&\leq &m,  
\end{eqnarray*}
warehouse $t(q)$ has room to receive at least $f-m \geq x^*_{q\eta}$ many products during the time interval corresponding to the room of $s(q)$. Therefore, by sending $x^*_{q\eta}$ portion of product $q$ during the time interval $[0,2m-1]$, 
we can obtain a feasible solution  
$x^{**}_{p\tt}$  such that 
$|\{ (p, \tt) \mid p \in P, \tt\geq 2m,  x^{**}_{p\tt} > 0\}|< \{ (p, \tt) \mid p \in P, \tt\geq 2m,  x^{*}_{p\tt} > 0\}|$. 
By repeatedly applying this procedure, we obtain a desired feasible solution of the linear relaxation  \eqref{eq-a1}--\eqref{eq-a5} with $T=T_{\min}$. 
\end{proof}

%%%%%%%%%%%%%%%%%%%%%%%%%%%

Finally we consider the case of $d^-\equiv +\infty$, since  the case of $d^+\equiv +\infty$ can be treated similarly.
In this case, we have no \indegree  capacity constraints,
and 
our feasibility problem can be represented by 
{\em the  generalized assignment problem}, that is, the problem to   assign  jobs in $N$ to machines in  $M_1$ (i.e., the formulation 
\eqref{eq-b0}, \eqref{eq-b1}, \eqref{eq-b3}, and  \eqref{eq-b4} with $|M_2|=1$). 
Therefore, the following result for the generalized assignment problem can be directly used  to obtain  Theorem \ref{bicriteria:th} (ii).  

\begin{theorem}[\cite{ST93}]
\label{GA:thm}
For the generalized assignment problem 
$($i.e., \eqref{eq-b0}, \eqref{eq-b1}, \eqref{eq-b3}, and  \eqref{eq-b4} with $|M_2|=1)$, 
there exists a polynomial-time
algorithm for finding an assignment of jobs in $N$ to
machines in $M_1$ whose assignment cost is at most the optimal, if capacity constraints for $i \in M_1$ 
 are augmented with  $s_{i}^{\max}$.
 \end{theorem}

%%%%%%%%%%%%%%%%%%

\section{Problem with No Warehouse Capacity Constraint}\label{nocap:sec}

In this section, we consider the
\mainpro with $\capa \equiv +\infty$,
i.e., the case where every \house
has a sufficiently large capacity.
As observed later in Theorem~\ref{hard2:th},
the \mainpro with $\capa \equiv +\infty$
is still strongly NP-hard.
On the other hand, in contrast to the
general cases, we can construct constant-factor algorithms for the problem without relaxing any constraint.
%In this section, we show the following constant-factor approximation algorithms
%for the \mainpro with $\capa \equiv +\infty$.

\begin{theorem}\label{apx:th}
$(i)$ The
\mainpro can be solved in $\mO(n+m\log{m})$ time
if  $\capa \equiv +\infty$ and 
 $d^- \equiv +\infty$ $($or $d^+ \equiv +\infty)$ and 
 %$\size$ 
 \pro size
 is uniform.

\noindent
$(ii)$  The
\mainpro is
 3/2-approximable 
 %in $\mO(m\log{m})$ time
if  $\capa \equiv +\infty$ and 
 $d^- \equiv +\infty$ $($or $d^+ \equiv +\infty)$ and $\transit$ is uniform.

\noindent
$(iii)$ The \mainpro is
 7/4-approximable 
 %in  $\mO(m\log{m})$ time 
if
$\capa \equiv +\infty$ and $d^- \equiv +\infty$ $($or $d^+ \equiv +\infty)$.

\ishii{
\noindent
$(iv)$  The \mainpro is
 4-approximable 
 %in polynomial time
if
$\capa \equiv +\infty$ and 
%$\size$
\pro size
is uniform.
}

\noindent
$(v)$  The \mainpro is
 6-approximable 
 %in polynomial time
if
$\capa \equiv +\infty$.
%\end{description}
\end{theorem}
We also show in  Theorem~\ref{binpack1:th} that the problem is inapproximable
within a ratio of $3/2-\varepsilon$ for any $\varepsilon>0$. 
This implies that the approximation ratio of Theorem \ref{apx:th} (ii) is {\em optimal}.

%\noi{\bf Case: $d^-=\infty$:}
\subsection{Case of $d^- \equiv + \infty$}\label{binpackalgo:subsec}
In this section, we  consider the case of $d^- \equiv+\infty$
and  prove  Theorem~\ref{apx:th}~(i)--(iii), 
where the case of $d^+
\equiv +\infty$
can be treated similarly.
Since $d^- \equiv+\infty$ and $\capa \equiv+\infty$,
we have only to consider a  schedule for $\P^+(w)$  independently
  for each \house $w
\in \W$.

In what follows, we
will show that 
  schedules for $\P^+(w)$ which
attain approximation ratios of
 Theorem~\ref{apx:th}\,(i)--(iii)
can be found in polynomial time for every $w \in \W$.

We consider a schedule for $\P^+(w)$;
namely, we consider an instance
$\I_{\rm \RA}=(\W',\P^+(w),$ 
$d^+(w),$
$\infty,\infty,\size,\transit)$ of the
\mainpro,
where $W'=\{w\} \cup \{\sink(p) \mid p \in \P^+(w)\}$ and we regard $\size$ and $\transit$ as those restricted to $\P^+(w)$.
Then, we need to partition $\P^+(w)$ 
into sets of \prods
whose total size is bounded by the \outdegree capacity $d^+(w)$.
Based on this observation,
we can see that the problem consisting of 
$\I_{\rm \RA}$'s has a similar structure
to
problem \textsc{Binpacking} defined below.
 We  construct approximation algorithms 
 corresponding to Theorem~\ref{apx:th}\,(i)--(iii)
 by using the ones  for \textsc{Binpacking}
 as  subroutines.

%we will reduce $\I_{\rm \RA}$ to an instance of  \textsc{Binpacking}, and  apply algorithms for \textsc{Binpacking} to the obtained instance.

\begin{quote}
%\begin{prob1}\label{VC-problem1}
\noindent
%{\rm {\sc 3-Partition  (3Part)}}
{\rm Problem \textsc{Binpacking}}

\noindent
%{\rm INSTANCE:}
Instance:  
$(I,\size_{BP},d):$  A set $I$ of
  items,
a function $\size_{BP}: I \rightarrow \mathbb{R}_+$, and
 a bin with capacity  $d \in \mathbb{R}_+$.

\noindent
%{\rm QUESTION:} 
Output: A packing of all items in $I$ with the minimum number of bins,
 i.e., 
a partition ${\cal J}$
of $I$ with the minimum $|{\cal J}|$
such that for every $J \in {\cal J}$, 
the total size of items in $J$ 
 is at most $d$.
%\qed
%\end{prob1}
\end{quote}

\noi
We construct from $\I_{\rm \RA}$ 
an instance
  $\I_{\rm BP}=(I,\size_{BP},d)$ of   \textsc{Binpacking} as follows.
For each \pro $p_i \in \P^+(w)$, we create an item
$i$ with $\size_{BP}(i)=\size(p_i)$; denote the resulting set of items by
$I$.
Let $d=d^+(w)$ as the capacity of a bin.
Then note that
a subset  $J$ of $I$ can be packed into one bin if and only if
the corresponding set $\{p_i \mid i \in J\}$ of \prods 
 can depart from $w$ simultaneously,
since the \outdegree capacity constraint for $w$ is satisfied. 
Hence, it is not difficult to see that
 $I$ can be packed into $k$ bins
if and only if any product in $\P$ can be sent from $w$  by  time  $k-1$, by mapping a set of items 
in the $\ell$th bin into a set of \prods departing from $w$ at time $\ell-1$.
%Moreover, if $\transit(p)=1$ for all $p \in \P^+(w)$, then
% $I$ can be packed into $k$ bins if and only if we can complete the reallocation of all \prods by time $k$.
%
%
Let  $\opt(w)$ be the minimum completion time of a schedule for
 $\I_{{\rm RA}}$
and $\opt_{{\rm BP}}(w)$ be the minimum number of bins for $\I_{{\rm BP}}$.
We then have the following inequality: 
\begin{equation}\label{binpack:eq}
 \opt(w) \ \geq \ \opt_{{\rm BP}}(w)-1+\min_{p \in \P^+(w)}\transit(p).
\end{equation}
\nop{
By these observations, we will 
obtain
constant-factor approximate solutions for  $\I_{{\rm RA}}$ 
based on those for   $\I_{{\rm BP}}$.
}

We first consider the case where $\transit$ is uniform, i.e., 
$\transit(p)=\transituni$ for all $p \in \P$.
It was shown in \cite{Simchi94}
that  the  so-called 
\textsc{First-Fit Decreasing (FFD)}  algorithm delivers
 in $\mO(|I|\log{|I|})$ time
a feasible solution ${\cal J}=\{J_1, \dots , J_k\}$
 for $\I_{\rm BP}$  with
 $k \leq \frac{3}{2}\opt_{{\rm BP}}(w)$,
where
 $J_\ell$ denotes the set of items packed in the $\ell$th bin
 for $\ell \in [1,k]$.
Let $\depart$ be the schedule for $\I_{{\rm RA}}$ 
such that $\depart(p)=\ell-1$ for every \pro $p \in \P^+(w)$ corresponding to an item in
$J_\ell$.
Its completion time is $k-1+\transituni \leq \frac{3}{2}\opt_{{\rm
BP}}(w)-1+\transituni$
$\leq \frac{3}{2}\opt(w)$ by (\ref{binpack:eq}).
This proves  Theorem~\ref{apx:th}\,(ii).

We next consider the case where $\transit$ is general.
We sort items in $I$ in such a way that
the corresponding  \prods  satisfy
 $\transit(p_1)\geq \transit(p_2) \geq
\cdots \geq \transit(p_{|I|})$.
%by nonincreasing order of $\transit$.
According to this order, we apply the so-called \textsc{First-Fit (FF)} algorithm 
to
 $\I_{\rm BP}$ 
to obtain a feasible solution
${\cal J}=\{J_1, \dots , J_k\}$
 for $\I_{\rm BP}$, 
where 
 $J_\ell$ denotes the set of items packed in the $\ell$th bin
 for $\ell\in [1,k]$.
Here, the \textsc{FF} algorithm packs each item, one by one, into the bin with the
lowest possible index, while opening a new bin if necessary.
It was shown in  \cite{Simchi94} that
$k \leq \frac{7}{4}\opt_{{\rm BP}}(w)$,
while $k=\opt_{{\rm BP}}(w)$ clearly holds when the $\size$ is uniform.
Define $\alpha$  by $1$ if the $\size$ is uniform,
and
$\frac{7}{4}$ otherwise.

Let  $\depart$ be the schedule for $\I_{{\rm RA}}$ 
such that $\depart(p)=\ell-1$ for every \pro $p \in \P^+(w)$ corresponding to an item in
$J_\ell$.
We then claim that the completion time  $\T_w$ for $\depart$ satisfies $\T_w \leq 
%\frac{7}{4}
\alpha
\opt(w)$. 
Since the time complexity of the algorithm is dominated by sorting
items in $I$, it can be implemented in $\mO(|I|\log{|I|})$ time.
Hence the following claim proves  Theorem~\ref{apx:th}\,(i) and (iii).

\begin{claim}
$\T_w \leq 
%\frac{7}{4}
\alpha
\opt(w)$.
\end{claim}
\begin{claimproof}
Let $i \in I$ be an item such that
the corresponding \pro $p \in \P^+(w)$ arrives at $\sink(p)$ at time 
$\T_w$.
Assume that 
%\kawahara{
$i$ is the $j$th item in $I$ and
%}
the \textsc{FF} algorithm puts $i$ in the $\ell$th bin.
%\kawahara{
Let $I_j$ be the set of the first $j$ items in $I$.
%}
By the assumption, we have 
$\T_w=\ell-1+\transit(p)$. 
Let $\depart_{i}$ be the schedule obtained from $\depart$ by restricting the product set $P$ to those corresponding to 
%$[1,i]$ 
$I_j$.
%\kawahara{(*** $\gets$ should be replaced with $I_j$)}. 
We can see that $\depart_{i}$ is the schedule obtained by the FF algorithm for 
%$[1,i]$ 
$I_j$
and $\T_w$ is also the completion time for $\depart_{i}$.  
%Thus,  $j_1-1+\transit(p_\ell)=\Q(w)$ holds also in ${\cal S}_\ell$.
Let 
 $\opt_{{\rm BP},i}$ be the optimal value for $\I_{{\rm BP}}$
restricted to 
%$[1,i]$, 
$I_j$,
and let   
 $\opt_{i}(w)$ be the optimal value for
 $\I_{{\rm RA}}$ restricted to the product set corresponding to 
 %$[1,i]$.
$I_j$.
Note that $\ell \leq 
%\frac{7}{4}
\alpha
\opt_{{\rm BP},i}$.
Since $I$ is sorted as above, 
(\ref{binpack:eq}) implies that 
$\opt_{i}(w) \geq \opt_{{\rm BP},i}-1+\transit(p)$.
\nop{
Thus, we have $\T_w=j_1-1+\transit(p_{i_1})$
$\leq \frac{7}{4}\opt_{{\rm BP},i_1}-1+\transit(p_{i_1})$
$\leq \frac{7}{4}(\opt_{{\rm BP},i_1}-1+\transit(p_{i_1}))$
$\leq \frac{7}{4}\opt_{i_1}(w)$
$\leq \frac{7}{4}\opt(w)$.
}
Therefore,  we have $\T_w=\ell-1+\transit(p)$
$\leq \alpha\opt_{{\rm BP},i}-1+\transit(p)$
$\leq \alpha(\opt_{{\rm BP},i}-1+\transit(p))$
$\leq \alpha\opt_{i}(w)$
$\leq \alpha\opt(w)$, which completes the proof of the claim. 
\end{claimproof}

\subsection{General cases}
%\noi{\bf General case:}
In this section, we prove Theorem~\ref{apx:th} (iv) and (v).
Let us first show Theorem~\ref{apx:th} (v) by  converting a
schedule $\depart$ for the capacity augmentation obtained by 
%the algorithm in Section~\ref{bicriteria:sec}
%Algorithm~\ref{iterative:algo}
algorithm \alg{Iterative}
%($\W,\P,d^+,d^-, \size,\transit,\T$)
into a 6-approximate schedule of the original reallocation problem. 
%For ${\cal S}^*(\T_{\min})$,
For this schedule $\depart$,
let
 $T_{\min}$,
$\P^+(w,\tt)$,  $\P^-(w,\tt)$, $p^+_i(w,\tt)$, and $p^-_i(w,\tt)$, $i=1,2$,
be defined as Section~\ref{bicriteria:sec}.
We first claim that a 9-approximate schedule can be easily obtained from $\depart$.

Let us  partition $P$ into three sets  $P_{1}=\{p \in P\mid p=p^+_1(s(p),\depart(p))\}$, 
$P_{2}=\{p \in P\mid p=p^+_2(s(p),\depart(p))\}$, and $P_{3}=P\setminus (P_1 \cup P_2)$. 
We further partition $P_{\alpha}$ ($\alpha=1,2,3$) into 3 sets $P_{\alpha,1}=\{p \in P_\alpha\mid p=p^-_1(t(p),\depart(p)+\tau(p))\}$, 
$P_{\alpha,2}=\{p \in P_\alpha\mid p=p^-_2(t(p),\depart(p)+\tau(p))\}$, and $P_{\alpha,3}=P_\alpha\setminus (P_{\alpha,1} \cup P_{\alpha,2})$. 
We construct a schedule $\psi$ such that 
$\psi(p)=(3\alpha+\beta-4)T_{\min} +\depart(p)$ if $p \in P_{\alpha,\beta}$. 
By definition, the schedule $\psi$ sends any product $p$ in $P_{\alpha,\beta}$  from $s(p)$ to $t(p)$ during the time interval $[(3\alpha+\beta-4)T_{\min},
(3\alpha+\beta-3)T_{\min}]$. 
By this, if 
two products $p$ and $p'$ satisfy $\psi(p)=\psi(p')$ or $\psi(p)+\tau(p)=\psi(p')+\tau(p')$, then they belong to the same set $P_{\alpha,\beta}$.
This together with 
%Theorem \ref{bicriteria:th} {\rm (i)} 
Lemma~\ref{additive:lem} 
implies that $\psi$ is feasible with the original reallocation problem.
Since the completion time for $\psi$ is at most $9\T_{\min}$, 
by 
Theorem \ref{bicriteria:th} {\rm (i)}, 
we can conclude that $\psi$ is a $9$-approximate 
feasible schedule. 

In order to improve the approximation ratio, 
we need a more careful treatment for modifying a schedule $\depart$ for the capacity augmentation given by Theorem \ref{bicriteria:th} {\rm (i)}. 
More precisely, we convert $\depart$  to a feasible
schedule with the completion time $6\T_{\min}$,
by giving the following three feasible schedules (a)--(c).
Here
for the schedule $\depart$,
 $T_{\min}$,
$\P^+(w,\tt)$ and $\P^-(w,\tt)$ are defined in Section~\ref{bicriteria:sec}. 

Products $p^+_i(w,\tt)$ and $p^-_i(w,\tt)$ ($i=1,2$)
are  defined  similarly as in Section~\ref{bicriteria:sec}. 
%\beforekawahara{However, we choose them in such a way that the set $Q_i=\{p^+_i(w,\tt), p^-_i(w,\tt) \mid w\in W, \tt \in [0,T_{\min}]\}$ is (inculsion-wise) minimal.}
%\kawahara{
Let $Q_i=\{p^+_i(w,\tt), p^-_i(w,\tt) \mid w\in W, \tt \in [0,T_{\min}]\}.$
If more than one product in $\P^+(w,\tt)$ (resp., $\P^-(w,\tt)$) has the same $i$th largest size, $Q_i$ is not determined uniquely.  In this case, we choose $p^+_i(w,\tt)$ (resp., $p^-_i(w,\tt)$) so that $Q_i$ is (inculsion-wise) minimal.
%}
Note that such $Q_1$ and $Q_2$ can be computed in polynomial time. 

%and for a product set $Q \subseteq P$, we define $\max Q$ by the set of products $p$ with the maximum size among $Q$, i.e., $\max Q=\arg\!\max \{ \size(p) \mid p \in Q\}$.  
\begin{alphaenumerate}
\item
A feasible schedule $\psi_1$ with
the completion time $3\T_{\min}$ for a product set $P_1=Q_1$.

\item A feasible schedule $\psi_2$  with
the completion time $2\T_{\min}$ for a product set $P_2=Q_2\setminus Q_1$. 
 
\item A feasible schedule $\psi_3$  with
the completion time $\T_{\min}$ for a product set $P_3=\P \sm (\P_1 \cup \P_2)$.
\end{alphaenumerate}
\nop{%%%%%%%%%%%
%%%%%%%%%%%%%
\begin{alphaenumerate}
\item
A feasible schedule $\psi_1$ with
the completion time $3\T_{\min}$ for a product set $\P_1 \subseteq \P$ which contains 
a \pro in $\max \P^+(w,\tt)$ 
and a \pro in $\max\P^-(w,\tt)$  
 for every  $w \in \W$
and $\tt \in [0,\T_{\min}]$.

\item A feasible schedule $\psi_2$  with
the completion time $2\T_{\min}$ for a product set $\P_2 \subseteq \P \sm \P_1$ which contains 
a \pro in $\max(\P^+(w,\tt)  \sm \P_1)$
and  a \pro in $\max(\P^-(w,\tt)  \sm \P_1)$
for every  $w \in \W$
and $\tt \in [0,\T_{\min}]$.
 
\item A feasible schedule $\psi_3$  with
the completion time $\T_{\min}$ for a product set $P_3=\P \sm (\P_1 \cup \P_2)$.
\end{alphaenumerate}
}
%%%%%%%%%%

Note that $\{P_1, P_2, P_3\}$ is a partition of $P$, and hence (a), (b), and (c) imply Theorem~\ref{apx:th} (v), 
since a desired  schedule $\psi^*$ can be obtained by $\psi^*(p)=\psi_1(p)$ if $p \in P_1$, $\psi_2(p)+3T_{\min}$ if  $p \in P_2$, and $\psi_3(p)+5T_{\min}$ if  $p \in P_3$. 
%We also note that $p_1^+(w,\tt)$ is contained in  $\max\P^+(w,\tt)$. 
%However, since $p_1^+(w,\tt)$'s are arbitrarily chosen from $\max\P^+(w,\tt)$'s, 
%no product set  might exist such that
%it contains $p_1^+(w,\tt)$'s and $p_1^-(w,\tt)$'s and has the minimum completion time at most $3T_{\min}$. 

In order to show these statements, let us construct an undirected bipartite graph $H_0=(\tilde{W}_1 \cup
 \tilde{W}_2,E_0)$ as follows. 
Recall that  $H=(\W_1\cup \W_2,E(H))$ is an  undirected bipartite graph
obtained from the \dmgraph $G=(\W,E_\P)$ defined  before Lemma~\ref{lemma-000a}.
%according to the proof of Lemma~\ref{matching:lem}.
For $i=1,2$, we replace every  $w \in \W_i$ 
with its $\T_{\min}+1$ copies 
 $w_{i,\tt}$,
$\tt\in [0,\T_{\min}]$; we denote  by $\tilde{W}_i$ the resulting set of vertices.
For every \pro $p \in \P$,
%For every edge $(w_1, w_2) \in E(G_D)$
%with $w_1=\source(p)$ and $w_2=\sink(p)$
% which corresponds to a \pro $p \in \P$, 
we 
replace the corresponding edge $(\source(p),\sink(p)) \in E(H)$ with 
an undirected edge
  $(\source(p)_{1,\depart(p)},
 \sink(p)_{2,\depart(p)+\transit(p)})$
which connects two vertices corresponding to
its departure and arrival time in
 $\depart$;
we denote by $E_0$ the resulting set of edges.
For simplicity, in the rest of this section, we identify  
products $p$ in $\P$ with edges  $e_p= (\source(p)_{1,\depart(p)},
 \sink(p)_{2,\depart(p)+\transit(p)}) \in E_0$. For example, we write  $\size(e)$ instead of $\size(p)$ if an edge $e$ corresponds to a product $p$.

For (a), we can see the following property on $P_1$.

\begin{lemma}
A graph  
 $H_1=(\tilde{W}_1 \cup \tilde{W}_2, P_1)$ is a forest. 
\end{lemma}
\begin{proof}
Assuming a contrary that $H_1$ contains 
a cycle $C$, we derive a contradiction.  We claim that 
all the edges in $C$ have the same size.
Let $V(C)=\{v_1,v_2,\ldots,v_k\}$, and we assume without loss of generality that $v_1=w_{1,\tt} \in \tilde{W}_1$ and $(v_1,v_2) =p^+_1(w,\tt)$.
Then we have $(v_2,v_3) =  p^-_1(w',\tt')$
for $v_2=w'_{2,\tt'} \in \tilde{W}_2$
%\kawahara{
because otherwise we can show that there exists an edge in $C$ that is neither $p^+_1(\hat{w},\hat{\tt})$ nor $p^-_1(\hat{w},\hat{\tt})$ for some $\hat{w} \in \tilde{W}_1 \cup \tilde{W}_2$ and $\hat{\tt} \in [0,\T_{\min}]$,
which contradicts the construction of $Q_1 (= P_1)$.
%}.
This means that $\size((v_1,v_2))\leq \size((v_2,v_3))$.
Similarly, we have $(v_3,v_4) =  p^+_1(w'',\tt'')$
for $v_3=w''_{1,\tt''} \in \tilde{W}_1$, which implies $\size((v_2,v_3))\leq \size((v_3,v_4))$.
By repeatedly applying this argument, 
we obtain 
 $\size((v_1,v_2)) \leq \size((v_2,v_3)) \leq
\dots \leq \size((v_k,v_1)) \leq \size((v_1,v_2))$, which proves the claim. 
However, this contradicts the minimality of $P_1$ 
%\kawahara{
because the set obtained by removing an edge in $C$ from $Q_1$ still satisfies the requirement of $Q_1$,
%}, 
a contradiction. 
\end{proof}

For a subset $F \subseteq E_0$, 
let $F(w)$ denote the set of edges in $F$ incident to $w$. 
Note that a set $F \subseteq E_0$  represents  a 
schedule with the completion time at most $T_{\min}$, and it is feasible 
if the total size of edges in $F(w)$
is at most $d^+(w)$ (resp., $d^-(w)$)
for every vertex $w \in \tilde{W}_1$ (resp., $w \in \tilde{W}_2$);
we call such an $F$  \emph{feasible}.
For (a), it suffices to show that $P_1$
can be partitioned into three feasible sets.

\begin{lemma}\label{forest:lem}
The set  $P_1$
can be partitioned into three feasible sets in polynomial time.  Hence we can compute in polynomial time  a feasible schedule of $P_1$ with the completion time at most 
$3\T_{\min}$. 
\end{lemma}
\begin{proof}
By Lemma~\ref{additive:lem}, we can observe that for every $w \in
\tilde{W}_1 \cup \tilde{W}_2$, 
 $P_1(w)$ can be partitioned into three
sets $R_i(w)$ ($i=1,2,3$)  such that
the total size of edges in $R_i(w)$ (namely, $\sum_{e \in R_i(w)}\size(e)$)
is at most $d^+(w)$ (resp., $d^-(w)$) for all $w \in \tilde{W}_1$ (resp., $w \in \tilde{W}_2$).  

Based on this, we prove the lemma by giving an algorithm
for partitioning $P_1$ into three feasible sets
$F_i$, $i=1,2,3$.
First we regard each component $X$ in $H_1$ as a rooted tree with root $r_X$
for a vertex $r_X \in V(X)$ chosen arbitrarily.
We initially let $F_i:=\emptyset$ for $i=1,2,3$, and
 repeat the following procedure for every vertex $w \in V(H_1)$ from the root to leaves
in a top-down way:
\begin{quote}
%\kawahara{
If $w = r_X$ for some $X$,
 update $F_i:=F_i \cup R_i(r_X)$ for $i=1,2,3$.
 Otherwise,
 %}
without loss of generality, 
assume that for the parent $v$ of $w$, $(v,w) \in R_1(w)$ 
and it is contained in the current $ F_1$. 
Update $F_i:=F_i \cup R_i(w)$ for $i=1,2,3$.
\end{quote} 
The resulting sets $F_1,F_2$, and $F_3$ are feasible, and can be computed in polynomial time. 
Moreover, by setting 
$\psi_1(p)=\depart(p)+(i-1)\T_{\min}$ if $p \in F_i$, we obtain a feasible schedule 
with the completion time at most 
$3\T_{\min}$. 
\end{proof}

For (b), 
let  $H_2=(\tilde{W}_1 \cup \tilde{W}_2, P_2)$.
Similarly to the discussion above, 
we can conclude that $H_2$ is a forest and 
$P_2$ can be partitioned into two feasible sets, since $P_2$ is disjoint from $Q_1$. 
\begin{lemma}\label{forest:lem-1}
The set  $P_2$
can be partitioned into two feasible sets in polynomial time.  Hence we can compute in polynomial time  a feasible schedule of $P_2$ with the completion time at most 
$2\T_{\min}$. 
\end{lemma}

As for %Case 
(c),  it is not difficult to see that $P_3$ itself is feasible by Lemma~\ref{additive:lem},
% great !! I have just realized 
%Theorem \ref{bicriteria:th} {\rm (i)},
since $P_3$ is disjoint from $Q_1$ and $Q_2$. 

\begin{lemma}\label{forest:lem-2}
We can compute in polynomial time  a feasible schedule of $P_3$ with the completion time at most 
$\T_{\min}$. 
\end{lemma}

From Lemmas  \ref{forest:lem},  \ref{forest:lem-1}, and \ref{forest:lem-2}, 
a 6-approximate feasible schedule of $P$ can be found 
in polynomial time, which proves   Theorem~\ref{apx:th}\,(v).

%%%

\medskip

Finally, 
we can show that
if the \pro size is uniform, then
 the \mainpro is 4-approximable 
 as shown in Lemma~\ref{4apx:lem},
which proves   Theorem~\ref{apx:th}~(iv).

\begin{lemma}\label{4apx:lem}
If the \pro size is uniform, then
$E_0$ can be partitioned into four feasible sets  in polynomial time.
\end{lemma}
\begin{proof}
Similarly to the discussion in Section~\ref{size2:sec}, we assume 
without loss of generality that $\size(p)=1$ for all \prods $p\in \P$,
and both of $d^+$ and $d^-$ are integral.
It follows from  Lemma~\ref{additive:lem} that 
\begin{equation}\label{4apx:eq}
\delta_{H_0}(w) \leq d^+(w)+2~ \ \ \ (\mbox{resp., }d^-(w)+2) 
\end{equation}
holds  for all $w \in \tilde{W}_1$ (resp., $\tilde{W}_2$). 
Here we recall that $\delta_{H_0}(w)$ is the degree of $w$ in $H_0$. 
In what follows, we construct four feasible sets $E_1$, $E_2$, $E_3$, and  $E_4$ which forms a partition of $E_0$. 

Let $E_1 \subseteq E_0$ be a maximal feasible set of edges,
and $H_1=(\tilde{W}_1 \cup \tilde{W}_2, E_0 \sm E_1)$.
We then claim that no two vertices of degree at least three are adjacent in   $H_1$.  
%for every vertex $v$ in $H'_1$ 
%with $\delta_{H_1'}(v)\geq 3$,
%with degree at least three,
%any neighbor of  $v$ has degree at most two.
Indeed, if $H_1$ would have an edge $(v_1,v_2)$ 
with $\delta_{H_1}(v_i)\geq 3$,
$i=1,2$, then  $E_1 \cup \{(v_1,v_2)\}$ would be feasible by
\eqref{4apx:eq},  contradicting the maximality of $E_1$.
Let $E_2$ 
be a  set of edges in $E_0 \sm E_1$ obtained 
 by arbitrarily choosing one edge from $(E_0 \sm E_1)(w)$
for each vertex $w$ with $\delta_{H_1}(w)\geq 3$.
%with degree at least three.
Then it follows from the claim that $E_2$ is a matching of $H_1$,  and hence $E_2$ is feasible. 
% and we have $d^+(v) \geq 1$ and $d^-(v)\geq 1$ for all $v \in \tilde{W}_1 \cup \tilde{W}_2$
%by assumption.
%Let  $H_2=(\tilde{W}_1 \cup \tilde{W}_2, E_0 \sm (E_1 \cup E_2))$.
%For proving the lemma,
%it suffices to
%show
%We will show 
Let us then partition 
$F=E_0 \sm (E_1 \cup E_2)$
into two
sets $E_3$ and $E_4$  such that
$E_3 \cap F(w), E_4 \cap F(w)\not=\emptyset$ for each vertex $w$ with $|F(w)|\geq 2$. 
Namely,  $F(w)$ is partitioned into two nonempty sets by $E_3$ and $E_4$ if it contains at least two edges.
%every vertex of degree at least two 
%in $H_2$ has both of 
%an edge in $E_3$ and an edge in $E_4$ incident to it.

We first show that such sets $E_3$ and $E_4$ are both feasible.
By symmetry, we only show the feasibility of $E_3$. 
Let $w$ be a vertex in $\tilde{W}_1$ 
 which is incident to  
 at least two edges in $E_3$. 
Then $E_4$ contains 
 at least one edge incident to $w$ by the definition of $E_3$ and $E_4$.
Also note that by  $|(E_0\setminus E_1)(w)|\geq |F(w)|\geq 3$,
$E_2$ contains an edge incident to $w$.
Hence we have $|E_3(w)|\leq \delta_{H_0}(w)-2$, which is at most $d^+(w)$ by  \eqref{4apx:eq}.
Similarly, we can see that 
$|E_3(w)|\leq d^-(w)$
for any $w \in \tilde{W}_2$. 
Therefore, $E_3$ is feasible.

We next show that such sets $E_3$ and $E_4$ can be found in the following manner:

\smallskip

\begin{description}
\item[(i)] Initialize $J=F$ and $F_3, F_4:=\emptyset$.

\item[(ii)] While $J$ contains a cycle $C=e_1,e_2, \dots, e_\ell$, 
update $F_3:=F_3 \cup \{e_{2k-1} \mid k \in [1,\ell/2]\}$,
$F_4:=F_4 \cup \{e_{2k} \mid k \in [1,\ell/2]\}$,
and $J:=J \sm C$.
%where 
%edges $e_1,e_2, \ldots, e_{\ell}$
%appear consecutively in $C$.
 
\item[(iii)]  For each component $X$ of the forest  $(\tilde{W}_1 \cup \tilde{W}_2,J)$, arbitrarily take a vertex
$r_X$ with degree one as a root of $X$, and regard $X$ as a rooted directed tree. 
Let $F_3'$  denote the set of edges in $F$ corresponding to directed edges
from $\tilde{W}_1$ to $\tilde{W}_2$ in the directed trees, and
let $F_4'=J\sm F_3'$.
Let $E_3 :=F_3 \cup F_3'$ and $E_4:=F_4 \cup  F_4'$.
\end{description}

\smallskip
\noi
Note that every cycle in (ii) consists of an even number of edges
since
$F$ forms a bipartite graph.
Hence if a vertex $w$ appears in some cycle in (ii), the property of 
$F_3 \cap F(w), F_4 \cap F(w)\not=\emptyset$ is satisfied. 
For the other vertices $w$ with $|F(w)|\geq 2$, let $J^*$ be the set $J$ in (iii). Then (iii) constructs sets $F_3'$ and $F_4'$ such that $F'_3 \cap J^*(w), F'_4 \cap J^*(w)\not=\emptyset$. 
Therefore, the resulting $E_3=F_3 \cup F'_3$ and $E_4=F_4 \cup F_4'$ satisfy the desired property.

Since all the sets $E_1$, $E_2$, $E_3$, and $E_4$ can be computed in polynomial time, the proof is completed. 
\end{proof}

\section{Intractability Results}\label{hardness:sec}

%%%%%%%%%%%%%%%%%%%

In this section, we investigate the intractability of the \mainpro.
In Section~\ref{hardness1:subsec},
we show that  deciding whether
the \mainpro with $|\W|=2$ or $\size\in\{1,2\}$
is feasible or not
is strongly {\rm NP}-complete.
It follows that the feasibility  is
para-NP-complete parameterized by $|\W|$
or the number of types of \prods.
In Section~\ref{hardness2:subsec},
we consider the case of
$\capa \equiv +\infty$, and
we show that even the case of $|\W|=2$ is 
strongly  NP-hard 
(and hence para-NP-hard parameterized by $|\W|$),
and that
the problem is inapproximable within a ratio
of $3/2-\varepsilon$ for any $\varepsilon>0$.
We also show that 
even the case of uniform \pro size is strongly NP-hard, in contrast to that
the case of uniform \pro size and transit time is polynomially solvable
as shown in Section~\ref{size2:sec}.

\subsection{General cases}\label{hardness1:subsec}

%We here show that even testing if \mainpro is
%feasible is strongly NP-hard.

%Let $\W$ and $\P$ be sets of \houses and \prods, respectively.
We first show that even if $|\W|=2$ and $d^-=+\infty$, 
 the feasibility  of the \mainpro is strongly NP-complete 
by a reduction from \textsc{3-Partition},
which is known to be strongly NP-complete \cite[p.224]{garey79npc}.

\med

\begin{quote}
%\begin{prob1}\label{VC-problem1}
\noindent
%{\rm {\sc 3-Partition  (3Part)}}
{\rm Problem \textsc{3-Partition}}

\noindent
%{\rm INSTANCE:}
Instance:  
$(\{x_1,x_2,\ldots,x_{3m}\},B):$  A set of
 $3m$ positive integers $x_1,x_2,\ldots,x_{3m}$ 
and an integer $B$
such that $\sum_{i \in [1,3m]}x_i=mB$ and $B/4 < x_i < B/2$ for each $i
\in [1,3m]$.

\noindent
%{\rm QUESTION:} 
Question: 
Is there  a partition $\{X_1, X_2, \ldots, X_m\}$ of
$[1,3m]$ such that $\sum_{i \in X_j}x_i$ $=B$ for each $j
\in [1,m]$?
%\qed
%\end{prob1}
\end{quote}

\begin{theorem}\label{hard1:th}
It is strongly NP-complete to decide whether 
  the \mainpro
is feasible 
even if we have $|\W|=2$,  $d^- \equiv +\infty$
 $($resp., $d^+ \equiv +\infty)$, and
$\transit \equiv 1$,
and  $d^+$  $($resp., $d^-)$ 
and $\capa$
 are both uniform.
\end{theorem}
\begin{proof}
We here show only the case where $d^+$ is uniform and $d^-\equiv +\infty$.
The case  where $d^-$ is uniform and $d^+\equiv +\infty$ can be treated
similarly. 
 Take an  instance $\I_{\rm 3PART}=(\{x_1,x_2,\ldots,x_{3m}\},B)$ 
of \textsc{3-Partition} such that $x_i$ is polynomial in $m$
for $i \in [1,3m]$.
From the $\I_{\rm 3PART}$, we construct an instance
$\I_{\rm \RA}=(\W,\P,d^+,d^-,\capa,\size,\transit)$ of the \mainpro as follows.
Let $\W=\{w_1,w_2\}$, $d^+(w_1)=d^+(w_2)=B$, 
 $d^-(w_1)=d^-(w_2)=+\infty$,
and
 $\capa(w_1)=\capa(w_2)=mB$.
Let
 $\P_1$  be the set of $3m$ \prods $p_i$, $i\in [1,3m]$,
such that every \pro $p_i \in \P_1$ satisfies $\source(p_i)=w_1$,
 $\sink(p_i)=w_2$, and $\size(p_i)=x_i$.
Let  $\P_2$  be the set of $m$ \prods %$p$
such that every \pro $p \in \P_2$ satisfies $\source(p)=w_2$,
 $\sink(p)=w_1$, and $\size(p)=B$, 
and $\P=\P_1\cup \P_2$.
Let $\transit(p)=1$ for all $p \in \P$.
Note that $\I_{\rm \RA}$ can be constructed from
 $\I_{\rm 3PART}$
in polynomial time.
Obviously, it is only possible to exchange 3 \prods in $\P_1$ with
the total size $B$ for 
one \pro in $\P_2$ at each time because $B/4 < \size(p_i) < B/2$
for every $i\in [1,3m]$.
It follows that there exists a feasible schedule for the \mainpro
if and only if  $\I_{\rm 3PART}$ is a yes-instance of  \textsc{3-Partition}.
Thus, the theorem is proved.
\end{proof}

%
%Moreover, we can show that even if $\size(p) \in \{1,2\}$ for all $p \in \P$,
%then the feasibility  of the \mainpro is strongly NP-complete
%in a similar way to \cite[Theorem 3.2]{miwa2000np}.

Moreover, we show that even if $\size(p) \in \{1,2\}$ for all $p \in \P$,
then the feasibility  of the \mainpro is strongly NP-complete.
Let $\I_{\rm \RA}=(\W,\P,d^+,d^-,\capa,\size, \transit)$
be  the instance
 of the \mainpro defined in the proof of Theorem~\ref{hard1:th}.
We will convert $\I_{\rm \RA}$ into an equivalent instance 
of the \mainpro with 
$\size\in \{1,2\}$ in a similar way to \cite[Theorem 3.2]{miwa2000np}.

Let $G=(\W,E_\P)$ be the \dmgraph of $\I_{\rm \RA}$.
Consider a directed edge $e=(u, v) \in E_\P$
corresponding to a \pro $p \in \P$
 where $\{u,v\}=\{w_1,w_2\}$;
note that if $u=w_1$ and $v=w_2$ (resp., $u=w_2$ and $v=w_1$),
then
   $p \in \P_1$ (resp., $\P_2$)  satisfies $\size(p)=x_i$ for some $i\in [1,3m]$
(resp.,  $\size(p)=B$).
We first create a set $U_p \cup V_p$ of $2x$ new vertices 
with $|U_p|=|V_p|=x$, where
let $x=\size(p)$  (note that $x~(=\size(p))$ is an integer).
Let $T_{u,p}$ (resp., $T_{v,p}$) 
be an in-tree (resp., out-tree) obtained by introducing some new 
vertices and directed edges
  so that
$U_p$ (resp., $V_p$) 
is the set of leaves and the in-degree (resp., out-degree) 
of every vertex not in $U_p$
(resp., $V_p$) is exactly two,
where a directed tree is called an \emph{in-tree} (resp., \emph{out-tree})
if the out-degree (resp., in-degree)
 of every vertex except its root is exactly one.
Note that such a tree $T_{u,p}$ (resp., $T_{v,p}$) can be constructed by
pairing two vertices with out-degree (resp., in-degree)
 zero from leaves to the root.
We denote the root of $T_{u,p}$ (resp., $T_{v,p}$) by $r_{u,p}$
(resp., $r_{v,p}$).
We then construct the graph $G_p$
% with $V(D_a)=\{u,v\}\cup V(T_{u,a}) \cup V(T_{v,a})$
in the following manner:
\begin{alphaenumerate}
\item We add a directed edge from $u$ to every vertex in $U_p$,
 a directed edge from every vertex in $V_p$ to $v$, and a directed edge 
$(r_{u,p},r_{v,p})$.

\item  We divide every vertex  $w \in V(T_{u,p}) \cup V(T_{v,p})
\sm(U_p \cup V_p \cup \{r_{u,p},r_{v,p}\})$
into two vertices $w'$ and $w''$,
 replace every directed edge entering $w$ with one entering $w'$,
replace every directed edge leaving $w$ with one leaving $w''$,
and add a directed edge $(w',w'')$.
\end{alphaenumerate}
For simplicity, we refer to a \pro corresponding to
a directed edge $e' \in E(G_p)$, its size, and
its transit time
as a \pro $e'$, the size of $e'$, 
and the transit time of $e'$, respectively. 
Let $W_{u,p}$ (resp.,  $W_{v,p}$)  denote the set of vertices generated
in (b) by dividing every vertex 
$w \in V(T_{u,p}) 
\sm(U_p \cup \{r_{u,p}\})$ 
(resp., $V(T_{v,p}) 
\sm(V_p \cup \{r_{v,p}\})$),
and $E'$ denote
 the set of the directed edges added in (b).
Let $\size(e')=2$ for every \pro $e'\in E' \cup \{(r_{u,p},r_{v,p})\}$
and   $\size(e')=1$ for all the other \prods $e' \in E(G_p) \sm (E' \cup \{(r_{u,p},r_{v,p})\})$.
%Let $\size(p')=2$ for every \pro $p'$ corresponding to
%an arc in $A' \cup \{(r_{u,a},r_{v,a})\}$
%and $\size(p')=1$ for every \pro  $p'$ corresponding to any other arc in
%$A(D_a)$.
Let $\capa(w)=2$ for  all $w \in W_{u,p} \cup W_{v,p}
 \cup \{r_{u,p},r_{v,p}\}$
and $\capa(w)=1$
 for all  $w \in U_p \cup V_p$.
Let $d^+(w)=B$ and $d^-(w)=\infty$ for all $w \in V(G_p)\sm \{u,v\}$.
Let $\transit(e)=1$ for all $e \in E(G_p)$.
Figure~\ref{size12:fig} shows an example of $G_p$ for a directed edge $e=(u,v)$
corresponding to $p \in \P$
with $\size(p)=5$.

\begin{figure}[t]
 \centering
 \includegraphics[width=0.9\textwidth]{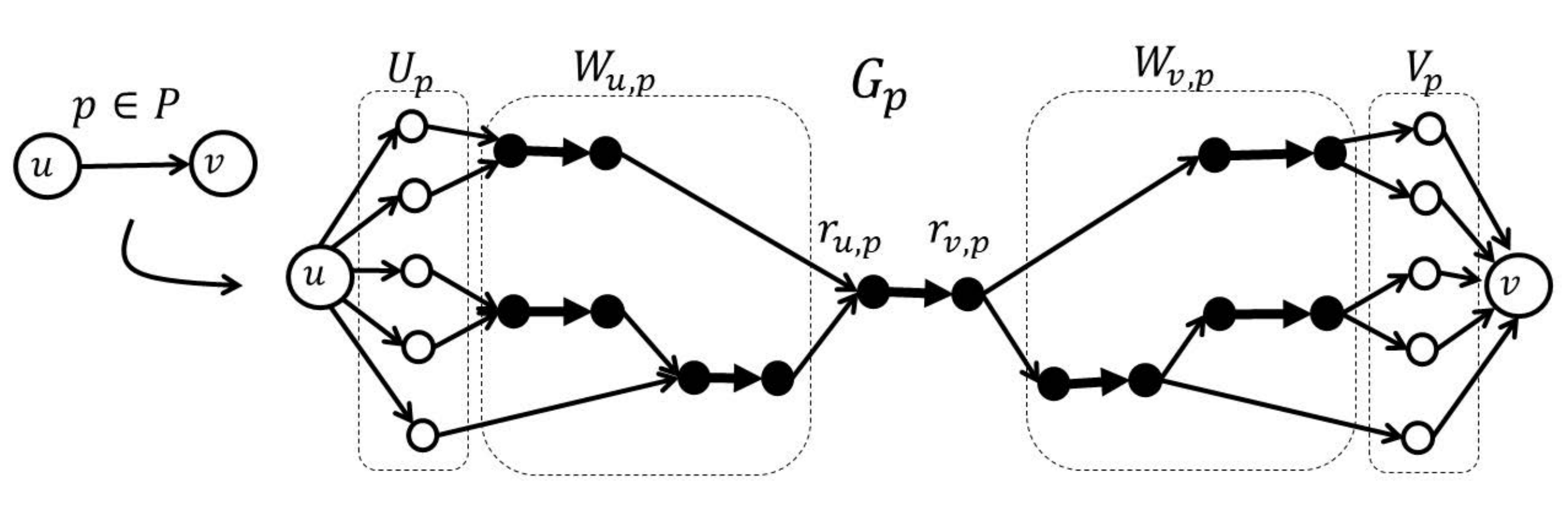}
\caption{Illustration of a graph $G_p$ for a
directed edge 
 $e=(u,v)$
corresponding to $p \in \P$
with $\size(p)=5$.
We have $\capa(w)=2$ for every vertex $w \in W_{u,p} \cup W_{v,p} \cup \{r_{u,p},r_{v,p}\}$,
 drawn as a black circle, while $\capa(w)=1$ for  all $w \in U_p \cup V_p$,
drawn as a white circle.
For every directed edge in $E' \cup \{(r_{u,p},r_{v,p})\}$, drawn as a bold arrow, 
its size  is two,
while the size of all the other directed edges is one.  
}
\label{size12:fig}
\end{figure}

The following lemma implies that
the graph $G_p$ plays the same role as $e=(u,v) \in E_\P$ corresponding to
$p$, where $E_u$ denotes the set of edges incident to vertices in $U_p \cup W_{u,p}   \cup \{u,r_{u,p}\}$.

\begin{lemma}\label{size12:lem}
If $v$ has  vacancy  at least $x$, then
all \prods in $E(G_p)$ can depart simultaneously. 
If $v$ has  vacancy  less than $x$, then
no \pro in 
%$E(G_p[U_p \cup W_{u,p}\cup \{u,r_{u,p},r_{v,p}\}])$
$E_u$
can depart.
\end{lemma}
\begin{proof}
The former case clearly holds.
Consider the latter case. 
Since $v$ has  vacancy less than $x$,
 there exists a vertex $w' \in V_p$ such that $(w',v)$ cannot depart.
Here notice that  from construction of $G_p$,
\begin{equation}\label{size12:eq}
 \begin{array}{lll}
&& \mbox{every \pro  $e'\in E(G_p) \sm \{(w,v)\mid w \in V_p\}$ can arrive at $\sink(e')$ only if 
  }\\
&& \mbox{all \prods initially located in $\sink(e')$ have departed.  }
 \end{array}
\end{equation}
Indeed,  (i) for every  $e' \in E'$, we have $\capa(\sink(e'))=2$
and (ii) for every  $e' \in E(G_p) \sm  (E' \cup \{(w,v)\mid w \in V_p\})$,
either $\capa(\sink(e'))=1$ holds
or  some \pro in $E'$ is initially located at $\sink(e')$
   and  $\capa(\sink(e'))=2$.
By (\ref{size12:eq}), every \pro in $E(G_p)$ on the path from $r_{u,p}$
to $w'$ cannot also depart.
Since $(r_{u,p},r_{v,p})$ cannot depart, it follows 
 again by (\ref{size12:eq}) that 
 no \pro in 
 %$E(G_p[U_p \cup W_{u,p}\cup \{u,r_{u,p}\}])$ 
 $E_u \sm \{(r_{u,p},r_{v,p})\}$
 can depart.
\end{proof}

Let $G'$ be the graph obtained from $G$ by replacing each edge in $E_\P$
corresponding to $p \in \P$
with $G_p$, and  $\I'_{\rm \RA}$ be the corresponding instance of the
\mainpro.
By Lemma~\ref{size12:lem}, we can observe that at each time,
 it is only possible to carry out \prods 
in $E(G_{p_1})\cup E(G_{p_2}) \cup E(G_{p_3})$
$ \cup E(G_{p_4})$ 
with 
   $p_i \in \P_1$ 
 for $i\in \{1,2,3\}$,
 $p_4\in \P_2$,
and $\sum_{i=1}^3\size(p_i)=\size(p_4)~(=B)$.
Thus,  $\I'_{\rm \RA}$ is equivalent to  $\I_{\rm \RA}$.
Note that the size of 
$\I'_{\rm \RA}$ is polynomial in $m$, since every $x_i$,
$i\in [1,3m]$, is
polynomial in $m$.
Hence, 
$\I'_{\rm \RA}$ can be constructed from $\I_{\rm \RA}$ in
polynomial time.
It follows that 
 the feasibility  of  $\I'_{\rm \RA}$ is also strongly NP-complete.
%,even if $\size(p) \in \{1,2\}$ for all $p \in \P$.

\begin{theorem}\label{size12:th}
It is strongly {\rm NP}-complete to decide whether the \mainpro is feasible
  even if we have   $d^- \equiv+\infty$
 $($resp., $d^+ \equiv +\infty)$,  
$\size \in \{1,2\}$, and 
$\transit\equiv 1$,
and  $d^+$  $($resp., $d^-)$
is
 uniform.
\end{theorem}

\noi
We 
finally 
remark  as a corollary of Theorems~\ref{hard1:th} and
\ref{size12:th} 
that
%deciding whether  the \mainpro
%is feasible is para-NP-complete
% parameterized by each of $|\W|$ and the number of types of \prods. 
%\
we have the following results about the
para-NP-completeness.

\begin{corollary}\label{paraNPhard1:cor}
Deciding whether  the \mainpro
is feasible is para-NP-complete
 parameterized by each of $|\W|$ and the number of types of \prods. 
\end{corollary}

%Also, we can show the following intractability results for the problem
%with no \house capacity constraints (i.e., $\capa\equiv +\infty$).

\subsection{Case of $\capa \equiv +\infty$}\label{hardness2:subsec}

We can observe that
even the case of $\capa \equiv \infty$  is
strongly NP-hard.
For two instances $\I_{\rm 3PART}$ of  \textsc{3-Partition}
 and $\I_{\rm \RA}$ of the \mainpro
in the proof of Theorem~\ref{hard1:th},
it is not difficult to see that 
there exists a schedule for  $\I_{\rm \RA}$ whose
  completion time is at most $m$ if and only if
$\I_{\rm 3PART}$ is a yes-instance of  \textsc{3-Partition}.
This follows since for completing the reallocation of all \prods
by time $m$,
it is only possible to exchange 3 \prods in $\P_1$ with
the total size $B$ for 
one \pro in $\P_2$ at each time.
Note that these arguments  need the \outdegree/\indegree  capacity
 constraints
but
not the \house capacity constraints.
Also, note that we can easily obtain a feasible solution for
$\I_{\rm \RA}$ since 
every \house has a sufficiently large capacity.
%we have no \house capacity constraints.
Hence, we have the following theorem.

\begin{theorem}\label{hard2:th}
The \mainpro
is strongly {\rm NP}-hard even if we have $|\W|=2$,  
$d^- \equiv +\infty$
 $($resp., $d^+ \equiv +\infty)$, 
$\capa \equiv +\infty$ and 
$\transit \equiv 1$
and  $d^+$  $($resp., $d^-)$
and $\capa$
 are both uniform.
Hence, it is para-{\rm NP}-hard parameterized by $|\W|$ even if $\capa \equiv +\infty$.
\end{theorem}

We next show the inapproximability of the problem 
by a reduction from \textsc{Binpacking},
which is known to be 
inapproximable within a ratio of $3/2-\varepsilon$ for any $\varepsilon>0$  (e.g., see
\cite{Simchi94}).

 Take an  instance $\I_{\rm BP}=(I=\{ i \in[1,|I|]\},\size_{BP}, d)$ 
 of \textsc{Binpacking}.
In  an opposite way to Section~\ref{binpackalgo:subsec},
 we construct from the $\I_{\rm BP}$ an instance
$\I_{\rm \RA}=(\W,\P,d^+,d^-,\infty,\size,\transit)$ of the \mainpro as follows.
Let $\W=\{u\} \cup \{w_i \mid i \in [1,|I|]\}$, $d^+(w)=d$ and 
 $d^-(w)=\infty$ for all $w \in \W$.
For every item $i \in I$, we create a \pro $p_i$
with $\source(p_i)=u$, $\sink(p_i)=w_i$, and
$\size(p_i)=\size_{BP}(i)$; denote the resulting set of \prods by $\P$.
Let $\transit(p)=1$ for all $p \in \P$.
%Then,     note that
%a subset  $I'$ of $I$ can be packed into one bin if and only if
%the corresponding set $\{p_i \mid i \in I'\}$ of \prods 
% can depart from $u$ simultaneously,
%since the \outdegree capacity constraints for $u$ are satisfied. 
%Hence, it is not difficult to see that
Similarly to the observations in Section~\ref{binpackalgo:subsec},
 $I$ can be packed into $k$ bins
if and only if the reallocation of all \prods in $\P$ can be
 completed  at  time  $k$.
 %, by mapping a set of products 
%in $k'$-th bin into a set of \prods departing from $u$ at time $k'-1$.
Thus, we have the following theorem, where we note
 that $\I_{\RA}$ can be constructed from
 $\I_{\rm BP}$
in polynomial time and that the case of $d^+ \equiv \infty$ can be treated similarly.

\begin{theorem}\label{binpack1:th}
The \mainpro
is inapproximable within a ratio
 of $3/2-\varepsilon$ for any $\varepsilon>0$ in polynomial time unless ${\rm P}={\rm NP}$, even in the case
 where
 $d^+$ is
 uniform,  $d^- \equiv +\infty$
  $($or $d^+ \equiv +\infty$,
   $d^-$ is uniform$)$,
 $c \equiv +\infty$,  and
$\transit \equiv 1$. 
\end{theorem}

%%%%%%%%%%%%%%ishii210922

We finally show that 
 the case of uniform \pro size is strongly NP-hard, in contrast to that
the case of uniform \pro size and transit time is polynomially solvable
as shown in Section~\ref{size2:sec}. 
Namely, we have the following theorem.

\begin{theorem}\label{flowshop:th}
The \mainpro
is strongly {\rm NP}-hard even if
 $|\W|=2$,
all of $d^+$, $d^-$, and 
%$\size$ 
\pro size
are uniform, 
and $\capa=+\infty$.
\end{theorem}

We prove this theorem by a reduction 
from the problem so-called \textsc{Two-Machine Flowshop with Delays (TMFD)}
 (e.g., see \cite{YHL04}).
In Problem \textsc{TMFD}, we are given two machines $M_1$ and $M_2$,
and a set $J$ of jobs. 
Every job $j \in J$ consists of two operations with an intermediate delay
$\ell_j \in \mathbb{Z}_+$;
the first (resp., second) operation is executed by $M_1$ (resp., $M_2$)
and
the time interval between the completion time of the first one
and the starting time of the second one is exactly $\ell_j$.
%where $\ell_j$ is a nonnegative integer.
Processing the first (resp., second) operation of  job $j$ takes
$p_{1j}$ (resp., $p_{2j}$), where $p_{ij}$ is a positive integer.
It follows that the completion time of job $j$ starting the first
operation
at time $\varphi(j)$ is $\varphi(j)+p_{1j}+\ell_j+p_{2j}$.
Each machine can process at most one job at any time.
The objective of \textsc{TMFD} is to find a schedule of all jobs in $J$
whose
completion time, i.e.,
$\max_{j \in J}\{\varphi(j)+p_{1j}+\ell_j+p_{2j}\}$ is minimized.
It was shown  that \textsc{TMFD} is strongly NP-hard
even if $p_{1j}=p_{2j}=1$ for all $j \in J$  \cite{YHL04}.

\begin{theorem}[\cite{YHL04}]
Problem  \textsc{TMFD} is strongly NP-hard
even if $p_{1j}=p_{2j}=1$ for all jobs $j \in J$.
\end{theorem}

 Take an  instance $\I_{\rm TMFD}=(M_1,M_2,J, \{\ell_j \mid j \in
 J\})$ 
of Problem \textsc{TMFD} 
such that $p_{1j}=p_{2j}=1$ for all $j \in J$ and
 each of $\ell_{j}$ is polynomial in $|J|$.
From the $\I_{\rm TMFD}$, we construct an instance
$\I_{\rm \RA}=(\W,\P,d^+,d^-,\capa,\size,\transit)$ of the \mainpro as follows.
Let $\W=\{w_1,w_2\}$, $d^+(w_1)=d^+(w_2)=d^-(w_1)=d^-(w_2)=1$, 
and
 $\capa(w_1)=\capa(w_2)=\infty$.
Let
 $\P$  be the set of  \prods $p_j$, $j\in [1,|J|]$,
such that every \pro $p_j \in \P$ satisfies $\source(p_j)=w_1$,
 $\sink(p_j)=w_2$,  $\size(p_j)=1$, and $\transit(p_j)=\ell_j+1$.
Note that $\I_{\rm \RA}$ can be constructed from
 $\I_{\rm TMFD}$
in polynomial time.
For proving Theorem~\ref{flowshop:th}, we will show that
there exists a schedule for  $\I_{\rm TMFD}$
whose completion time is at most $T$
if and only if
there exists a schedule for  $\I_{\rm \RA}$
whose completion time is at most $T-1$.

Assume that there exists a schedule $\depart'$ for  $\I_{\rm TMFD}$
whose completion time is at most $T$;
let $\depart'(j)$ denote the time when job $j \in J$ starts 
the first operation in the schedule $\depart'$.
Then, job $j$ starts the second operation at time
$\depart'(j)+\ell_j+1$
by $p_{1j}=1$.
Since each machine can process at most one job at any time,
we have
\begin{equation}\label{flowshop1:eq}
 \depart'(j)\neq \depart'(j') \mbox{ and }
\depart'(j)+\ell_j+1 \neq \depart'(j')+\ell_{j'}+1
\end{equation}
for every two distinct jobs $j,j' \in J$. 
Note that $\max_{j \in J}\{\depart'(j)+\ell_j+2\} \leq T$ by $p_{2j}=1$.
Let $\depart$ be the schedule for  $\I_{\rm \RA}$
such that $\depart(p_j)=\depart'(j)$ for $p_j \in \P$.
Then, \pro $p_j$ arrives at $w_2$ at time $\depart'(j)+\ell_j+1$
by $\transit(p_j)=\ell_j+1$.
By \eqref{flowshop1:eq} and $\size(p_j)=1$, 
 $\depart$ satisfies the \outdegree and \indegree capacity constraints.
By $\capa(w_1)=\capa(w_2)=\infty$, it follows that $\depart$ is feasible.
The completion time for $\depart$ is
$\max_{p_j \in \P}\{\depart'(j)+\ell_j+1\}\leq T-1$.

Assume that
 there exists a schedule $\depart$ for  $\I_{\rm \RA}$
whose completion time is at most $T-1$.
Then since $d^+(w_1)=d^-(w_2)=1$ and $\size(p_j)=1$ and $\transit(p_j)=\ell_j+1$ for  $p_j\in \P$,
it follows by
the \outdegree and \indegree capacity constraints 
that
\begin{equation}\label{flowshop2:eq}
 \depart(p_j)\neq \depart(p_{j'}) \mbox{ and }
\depart(p_j)+\ell_j+1 \neq \depart(p_{j'})+\ell_{j'}+1
\end{equation}
for every two distinct \prods $p_j,p_{j'} \in \P$. 
Note that $\max_{p_j \in \P}\{\depart(p_j)+\ell_j+1\}\leq T-1$.
Let $\depart'$ be the schedule for  $\I_{\rm TMFD}$
such that  job $j \in J$ starts 
the first operation at time $\depart(p_j)$.
Then, job $j$ starts the second operation at time
$\depart(p_j)+\ell_j+1$
and completes it at time 
$\depart(p_j)+\ell_j+2$.
Since each machine processes at most one job at any time in $\depart'$
by \eqref{flowshop2:eq} and $p_{1j}=p_{2j}=1$  for all $j \in J$, 
it follows that $\depart'$ is a feasible schedule for 
 $\I_{\rm TMFD}$.
The completion time for $\depart'$ is
$\max_{j \in J}\{\depart(p_j)+\ell_j+2\} \leq T$.

\bibliography{./mobility}

\end{document}